\documentclass[journal]{IEEEtran}

\usepackage{amsmath}
\usepackage{amsfonts}
\usepackage{amssymb}
\usepackage{amsthm}
\usepackage{mathrsfs} 

\usepackage{url}
\usepackage[colorlinks]{hyperref}

\newtheorem{proposition}{Proposition}
\newtheorem{lemma}{Lemma}

\usepackage{graphicx}
\usepackage{xcolor}
\usepackage{subfigure}

\makeatletter
\def\blfootnote{\xdef\@thefnmark{}\@footnotetext}
\makeatother

\begin{document}

\title{Enormous Fluid Antenna Systems (E-FAS) for Wireless Sensing: Channel Modeling and Conditional Estimation Limits} 

\author{Farshad~Rostami~Ghadi,~\IEEEmembership{Senior Member},~\textit{IEEE}, 
Kai-Kit~Wong,~\IEEEmembership{Fellow},~\textit{IEEE},\\ 
Jos\'e D. Vega-S\'anchez,~\IEEEmembership{Senior Member},~\textit{IEEE}, 
Kin-Fai Tong,~\IEEEmembership{Fellow},~\textit{IEEE}, and 
Hyundong Shin,~\IEEEmembership{Fellow,~IEEE}
\vspace{-5mm}
}

\maketitle

\begin{abstract}
In this paper, we develop a fundamental analytical framework for integrated sensing and communications (ISAC) enabled by the Enormous Fluid Antenna System (E-FAS), which transforms a collection of coordinated intelligent surfaces into a gigantic reconfigurable electromagnetic aperture, with particular emphasis on the limits of angular sensing. We begin by developing a bidirectional sensing channel model that explicitly captures the complete sensing process, including surface-wave (SW) routing, distributed reradiation, target scattering, and echo propagation. Based on this channel model, we formulate a parametric observation model for target sensing and derive the associated Fisher information matrix (FIM) and Cram\'{e}r-Rao bound (CRB) for angular estimation. The analysis demonstrates that E-FAS gives rise to a fundamentally different sensing regime compared with conventional array-based and reconfigurable-surface-aided ISAC architectures. Our analysis uncovers that maximizing coherent routing gain does not necessarily maximize sensing performance, exposing a fundamental trade-off between SW routing gain and sensing diversity in programmable propagation environments. Numerical results validate the developed framework and demonstrate that E-FAS-enabled ISAC systems can achieve substantial angular sensing gains over conventional architectures under the same transmit-power budget. The results further underscore the importance of jointly optimizing propagation routing and sensing functionality, positioning E-FAS as a new paradigm for ISAC.
\end{abstract}


\blfootnote{The work of F. Rostami Ghadi and K. K. Wong is supported by the Engineering and Physical Sciences Research Council (EPSRC) under Grant EP/W026813/1. The work of J. D. Vega-S$\acute{\rm a}$nchez was supported in part by the Universidad San Francisco de Quito through the Poli-Grants Program under Grant 41989 and Grant 30502. The work of H. Shin is supported by the National Research Foundation of Korea (NRF) grant funded by the Korean government (MSIT) (RS-2025-00556064 and RS-2025-25442355), and by the Ministry of Science and ICT (MSIT), Korea, under the ITRC (Information Technology Research Center) support program (IITP-2025-RS-2021-II212046), supervised by the IITP (Institute for Information \& Communications Technology Planning \& Evaluation).}

\blfootnote{\noindent F. Rostami Ghadi and K. K. Wong are with the Department of Electronic and Electrical Engineering, University College London, WC1E 7JE London, United Kingdom. K. K. Wong is also affiliated with the Department of Electronic Engineering, Kyung Hee University, Yongin-si, Gyeonggi-do 17104, Republic of Korea (e-mail: {f.rostamighadi, kai-kit.wong}@ucl.ac.uk).}
\blfootnote{J. D. Vega-S$\acute{\rm a}$nchez is with the Colegio de Ciencias e Ingenier\'{i}as ``El Polit\'{e}cnico'' Universidad San Francisco de Quito (USFQ), Quito 170157, Ecuador (email: dvega@usfq.edu.ec).}
\blfootnote{K.-F. Tong is with the School of Science and Technology, Hong Kong Metropolitan University, Hong Kong SAR, China (e-mail: ktong@hkmu.edu.hk).}
\blfootnote{H. Shin is with the Department of Electronics and Information Convergence Engineering, Kyung Hee University, Yongin-si, Gyeonggi-do 17104, Republic of Korea (e-mail: hshin@khu.ac.kr).}

\blfootnote{Corresponding Author: Kai-Kit Wong.}

\vspace{-2mm}
\section{Introduction}\label{sec-intro}
\IEEEPARstart{I}{ntegrated} sensing and communications (ISAC) has emerged as a key enabling technology for next-generation networks, a.k.a.~sixth-generation (6G), where communication and sensing functionalities are realized over shared hardware, spectrum, and signaling resources \cite{liu2022integrated}. By jointly designing transmission strategies, ISAC systems aim to improve spectral efficiency while enabling accurate environmental perception, with applications in autonomous systems, smart environments, and next-generation mobile networks \cite{wei2022toward,wei2023isac}. Consequently, significant research efforts have been devoted to developing signal processing, waveform design, and performance analysis frameworks for ISAC systems. Existing ISAC architectures are predominantly based on conventional antenna arrays \cite{liu2020joint} or more recently, on reconfigurable intelligent surfaces (RISs) \cite{Yu2024active,zhang2024joint,chen2024sim}. While these approaches provide important degrees of freedom (DoF) for beamforming and propagation control, they remain fundamentally limited by the underlying propagation environment \cite{di2020smart}. In particular, signal attenuation, and limited angular diversity constrain both communication performance and sensing accuracy, especially in complex scenarios.

Recently, fluid antenna systems (FASs) have emerged as a reconfigurable physical-layer architecture for adaptive wireless propagation control, where antenna positions or effective radiation locations can be dynamically reconfigured to exploit favorable channel conditions \cite{new2025tutorial,new2025flar,Lu-2025,hong,wu2024flu}. Unlike conventional fixed-array architectures, FAS-enabled systems introduce an additional spatial DoF through adaptive antenna positioning, thereby improving channel diversity, interference management, and spatial beamforming capability \cite{Wong2023FAS,wong2020performe}. Motivated by these advantages, recent studies have extended FAS concepts to ISAC systems, demonstrating the potential of fluid antenna architectures for improving sensing resolution and communication reliability in dynamic propagation environments \cite{wang2024fluid,zhou2024fluid,zou2024shifting,ghadi2025performISAC,zhang2026fund,ghadi2026ai}. In parallel, FAS assisting RIS \cite{ghadi2024on,tang2025power,yao2026fram} and reconfigurable electromagnetic architectures \cite{xiao2025fluid,ghadi2025fris,ghadi2025fires,salem2025first,kaveh2026physical,xiao2025from} have also attracted increasing attention as a means of enhancing propagation adaptability and spatial coverage. These approaches aim to combine the propagation control capability of RISs with the spatial flexibility of fluid  antenna structures. Nevertheless, most existing FAS and RIS architectures still rely primarily on conventional propagation mechanisms and remain sensitive to blockage, attenuation, and limited spatial diversity.

Inspired by the emerging concept of fluid and reconfigurable propagation engineering, enormous fluid antenna systems (E-FASs) have recently been proposed as a paradigm shift for large-scale manipulation of wireless propagation, extending the use of FAS in a large scale \cite{wong2021vision,wong2025enormous,wu-efas2026}. Specifically, an E-FAS is an ensemble of one or more coordinated intelligent surfaces that helps establish communication between the transmitter(s) and receiver(s). Unlike conventional RIS architectures that primarily rely on anomalous reflection, E-FAS enables electromagnetic waves to be coupled into guided surface wave (SW) modes that propagate along engineered quasi two-dimensional structures before being reradiated toward intended users. Hence, the metasurfaces act as programmable electromagnetic interfaces that support controlled energy transport and spatial wave manipulation. Since the guided SWs experience cylindrical rather than spherical spreading, propagation attenuation can be significantly reduced compared to free-space transmission \cite{Liu-access2024}. This mechanism enables signals to be routed around obstacles, redirected along walls or ceilings, and delivered to desired regions with enhanced efficiency and coverage. Existing studies on E-FAS have demonstrated substantial gains in communication performance, energy efficiency, and new multiuser transmission capabilities.

Despite the initial vision of SW communication in \cite{wong2021vision}, many fundamental aspects of E-FAS design and performance remain unexplored. Early discussions in \cite{wong2021vision} introduced the concept of guiding electromagnetic waves along engineered surfaces to improve propagation efficiency, interference control and coverage. More recently, \cite{wong2025enormous,wu-efas2026} presented the broader E-FAS architecture and highlighted its potential for future wireless systems. Initial analytical studies in \cite{ghadi2026enormous} demonstrated promising gains in outage probability, ergodic rate, and multiuser transmission capability through SW-assisted communication. Furthermore, channel estimation aspects of E-FAS were investigated in \cite{ghadi2026efas1}, while the impact of correlated SW leakage on physical layer security was addressed in \cite{ghadi2026e-fas}. These early results collectively suggest that E-FAS fundamentally reshape wireless propagation and transmission design.

\vspace{-2mm}
\subsection{Motivation and Contributions}
Despite the aforesaid research advances, the potential of E-FAS for ISAC systems remains largely unexplored. Existing studies primarily focus on communication performance benefits and do not investigate how guided SW routing and spatially distributed reradiation influence sensing functionality. Meanwhile, conventional ISAC models are typically developed for direct propagation or reflection-based architectures and therefore do not capture the bidirectional and routing-dependent propagation structure induced by E-FAS. As a consequence, the fundamental sensing behavior of E-FAS-enabled ISAC systems, including the impact of routing gain, routing correlation, and radiation geometry on estimation performance, remains poorly understood. Addressing this challenge requires a new analytical framework beyond conventional array sensing models. In E-FAS-enabled systems, the sensing channel is inherently composite and is formed through the interaction of guided SW propagation, distributed reradiation, target scattering, and return echo routing. As a result, the effective sensing manifold depends not only on the physical array structure, but also on the underlying routing configuration of the E-FAS. This introduces new sensing characteristics and new trade-offs that do not arise in conventional ISAC architectures.

In this paper, we develop an estimation-theoretic framework for E-FAS-enabled ISAC systems with emphasis on the fundamental limits of angular sensing. The main contributions of this work are summarized as follows:
\begin{itemize}
\item First, we develop a conditional end-to-end sensing channel model for E-FAS-enabled ISAC systems that captures guided SW routing, distributed reradiation, target  scattering, and echo reception. The proposed formulation explicitly characterizes the forward and return propagation processes and therefore establishes a routing-dependent sensing structure which is fundamentally different from conventional array and reflection ISAC models.
\item We formulate a conditional parametric observation model for angular estimation and derive the Fisher information and Cram\'{e}r-Rao bound (CRB), from which we uncover how the sensing performance depends jointly on routing coefficients, propagation geometry, and channel structure.
\item We characterize the effect of routing-induced correlation on the effective sensing dimensionality of the E-FAS architecture. In particular, we establish sufficient conditions under which strong routing alignment and channel correlation lead to rank collapse of the sensing channel and reduce the number of effective sensing modes.
\item Moreover, we derive scaling relationships that connect sensing performance with SW routing gain, spatial geometry, and routing-induced correlation. The analysis shows that increasing routing concentration can improve signal amplification while simultaneously reducing spatial diversity, thereby revealing a fundamental gain-diversity trade-off unique to E-FAS-enabled sensing systems.
\item Our numerical results validate our analytical framework and illustrate the impact of routing configuration and radiation geometry on sensing accuracy. The results further illustrate the potential advantages of E-FAS-enabled ISAC architectures compared with conventional propagation structures under comparable transmit power conditions.
\end{itemize}


\vspace{-2mm}
\section{System Model}\label{sec:sys}
As shown in Fig.~\ref{fig:system}, we consider a monostatic ISAC system in which a base station (BS) is equipped with $N_{\mathrm{t}}$ fixed transmit antennas and $N_{\mathrm{r}}$ fixed receive antennas. The BS performs joint transmission and reception over a shared propagation environment assisted by an E-FAS infrastructure. The E-FAS consists of coordinated electromagnetic surfaces that support incident wave capture, guided SW propagation, and controlled reradiation into free-space at selected locations. Consequently, the end-to-end propagation between the BS and external nodes is realized through a combination of guided SW routing and short-range free-space transmission \cite{ghadi2026enormous}. The E-FAS-induced reradiation points are modeled as a finite set of effective radiation locations. These locations represent equivalent emission points determined by the underlying routing configuration and do not correspond to independent physical antenna elements. Instead, they provide an analytical abstraction of the spatial radiation structure generated by the E-FAS-assisted propagation environment. In particular, we consider a single point target and a single communication user. This setting isolates the fundamental sensing structure induced by the E-FAS architecture while maintaining analytical tractability.\vspace{-4mm}

\begin{figure}[t]
\centering
\includegraphics[width=.85\columnwidth]{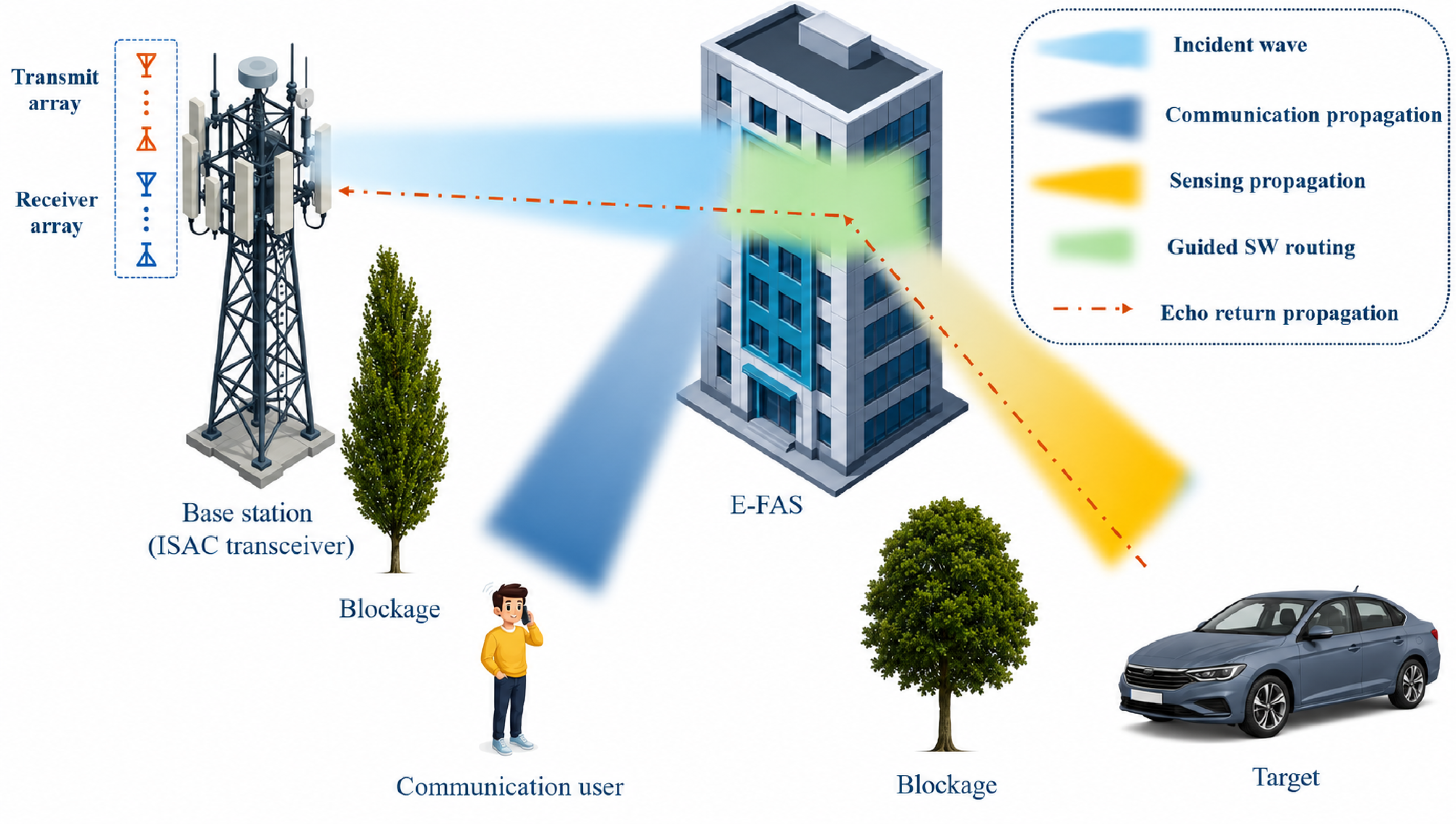}
\caption{Illustration of the E-FAS-enabled ISAC architecture.}\label{fig:system}
\vspace{-6mm}
\end{figure}

\vspace{-2mm}
\subsection{Geometric Model}
We consider a Cartesian coordinate system and denote the position of the BS by $\mathbf{p}_{\mathrm{B}} \in \mathbb{R}^{3}$, the point target by $\mathbf{p}_{0} \in \mathbb{R}^{3}$, and the communication user by $\mathbf{p}_{\mathrm{u}} \in \mathbb{R}^{3}$. The E-FAS infrastructure induces a finite set of effective radiation points indexed by $\mathcal{M} \triangleq \{1,2,\ldots,M\}$, with corresponding positions $\mathbf{p}_{m} \in \mathbb{R}^{3}$, $m \in \mathcal{M}$. These radiation points represent equivalent emission locations through which the guided SW energy is coupled back into free space. For each radiation point $m$, we define the geometric distances from the point to the target and from the point to the user, respectively, as
\begin{align}
d_{m,0} &\triangleq \|\mathbf{p}_{m}-\mathbf{p}_{0}\|_{2}, \label{eq-dism}\\
d_{m,\mathrm{u}} &\triangleq \|\mathbf{p}_{m}-\mathbf{p}_{\mathrm{u}}\|_{2}.
\end{align}
The positions of the radiation points are assumed to be known through infrastructure design or prior calibration, such that the unknown parameters are associated with the target location.

We also define the target parameter vector as 
\begin{align}
\boldsymbol{\eta} \triangleq [\tau_{0}, \theta_{0}]^{T},
\end{align} 
where $\tau_{0}$ is the round trip propagation delay and $\theta_{0}$ defines the angular parameter of interest \cite{trees2002optimum,stoica1989music}. The extension to vector-valued angle parameters is straightforward. For the angular sensing analysis later, the target angle $\theta_0$ is defined with respect to a reference coordinate system centered at the BS. In particular, the target position can be expressed as
\begin{align}
\mathbf p_0=\mathbf p_{\mathrm B}+r_0\mathbf u(\theta_0),
\end{align}
where $r_0$ denotes the target range and $\mathbf u(\theta_0)$ is the corresponding unit-direction vector. Therefore, the propagation distances, $d_{m,0}$, and the free-space coefficients, $\ell_{m,0}(\theta_0)$, inherit an implicit dependence on $\theta_0$ through the target location $\mathbf p_0$.

\vspace{-2mm}
\subsection{Routing State Representation}
The propagation through the E-FAS is characterized by a routing state that captures the effective transfer from the BS to each radiation point through guided SW paths \cite{wong2025enormous,ghadi2026enormous}. We define the routing state vector as
\begin{align}
\boldsymbol{\rho} \triangleq [\rho_{1},\rho_{2},\ldots,\rho_{M}]^{T}, \label{eq-rho_v}
\end{align}
where $\rho_{m} \in \mathbb{C}$ is the effective routing coefficient associated with the $m$-th radiation point. In the narrowband model adopted for the main analysis, we write $\rho_{m}$ as
\begin{align}
\rho_{m} = \sqrt{G_{m}^{(\mathrm{sw})}} e^{-j\phi_{m}^{(\mathrm{sw})}},\label{eq-rho_m}
\end{align}
where $G_{m}^{(\mathrm{sw})}\ge 0$ is the effective guided wave power transfer factor and $\phi_{m}^{(\mathrm{sw})}$ is the
corresponding accumulated phase.

The coefficient $\rho_m$ is an equivalent composite quantity that captures the net effect of SW guided propagation, attenuation, junction loss, and coupling along the underlying SW route. As a result, \eqref{eq-rho_m} is interpreted as a reduced order representation of a potentially multi-segment routing process rather than as a freely adjustable reflection coefficient. Hence, to connect this abstraction with the underlying propagation physics, a frequency-dependent refinement can be introduced as
\begin{align}
\rho_m(f)=\sqrt{G_m(f)}\,e^{-j(2\pi f\tau_m+\phi_m^{(c)})},\label{eq:rho_f}
\end{align}
where $\tau_m$ denotes the effective guided wave propagation delay, $\phi_m^{(c)}$ denotes an additional coupling-induced phase offset, and $G_m(f)$ is the frequency-dependent attenuation factor which can be physically modeled as
\begin{align}
G_m(f)=\exp\!\big(-\alpha(f)L_m\big),\label{eq:Gm_f}
\end{align}
where $L_m$ is the effective routing path length and $\alpha(f)$ is the surface attenuation coefficient \cite{pozar2011micro}. The narrowband model in \eqref{eq-rho_m} is recovered by evaluating \eqref{eq:rho_f} at the carrier frequency. Throughout this paper, the routing state vector $\boldsymbol{\rho}$ is assumed quasi-static over the sensing interval, and the analysis is conducted conditioned on a given routing realization. Nevertheless, when analyzing structural effects such as routing-induced coupling and mode correlation, the routing coefficients may also be interpreted as realizations drawn from an ensemble of feasible routing configurations. This ensemble viewpoint is introduced only for second order statistical characterization and does not alter the conditional observation model.

\vspace{-2mm}
\subsection{Routing-Induced Correlation}
As different guided routes may partially share common SW segments, the routing coefficients may exhibit statistical dependence across different radiation points \cite{ghadi2026e-fas}. To characterize this effect, we introduce the routing correlation matrix
\begin{align}
\mathbf{C}_{\rho}\triangleq\mathbb{E}\!\left[\boldsymbol{\rho}\boldsymbol{\rho}^{H}\right],
\end{align}
where the expectation is taken over an ensemble of feasible routing configurations, coupling conditions, and environmental perturbations associated with different operating states of the E-FAS infrastructure.

Specifically, the $(m,n)$-th entry of $\mathbf{C}_{\rho}$ is expressed as
\begin{align}
c_{m,n}=\sqrt{\gamma_m\gamma_n}\,\xi_{m,n},~|\xi_{m,n}| \le 1,
\end{align}
where $\gamma_m \triangleq \mathbb{E}[|\rho_m|^2]$ represents the average routing power associated with the $m$-th radiation point, and $\xi_{m,n}$ denotes the normalized routing correlation coefficient between the $m$-th and $n$-th routing paths.

\vspace{-2mm}
\subsection{Correlated Channel Model}
We suppose that $\mathbf{h}_{\mathrm{B},m}\in\mathbb{C}^{N_{\mathrm{t}}\times 1}$ denotes the effective BS-to-surface coupling channel associated with the route terminating at the $m$-th radiation point. Also, it is assumed that $\mathbf{g}_{m,\mathrm{R}}\in\mathbb{C}^{N_{\mathrm{r}}\times 1}$ defines the corresponding channel at the receive side from the $m$-th radiation point through the E-FAS return path to the BS. To maintain analytical tractability, the dependence among different routes is modeled at the level of second-order statistics \cite{shiu2000fading,kermoal2002stochastic}. In particular, we assume
\begin{align}
\mathbb{E}\!\left[\mathbf{h}_{\mathrm{B},m}\mathbf{h}_{\mathrm{B},n}^{H}\right]&=\beta^{(\mathrm{t})}_{m,n}\mathbf{R}_{\mathrm{t}},
\label{eq:tx_corr_model}\\
\mathbb{E}\!\left[\mathbf{g}_{m,\mathrm{R}}\mathbf{g}_{n,\mathrm{R}}^{H}\right]&=\beta^{(\mathrm{r})}_{m,n}\mathbf{R}_{\mathrm{r}},\label{eq:rx_corr_model}
\end{align}
in which $\mathbf{R}_{\mathrm{t}} \succeq \mathbf{0}$ and $\mathbf{R}_{\mathrm{r}} \succeq \mathbf{0}$ denote the BS transmit- and receive-side spatial covariance matrices, respectively, while $\beta^{(\mathrm{t})}_{m,n}$ and $\beta^{(\mathrm{r})}_{m,n}$ quantify
routing-induced dependence across different paths. By construction, we have
\begin{align}
\beta^{(\mathrm{t})}_{m,m}=\beta^{(\mathrm{r})}_{m,m}=1,~|\beta^{(\mathrm{t})}_{m,n}|\le 1,\;|\beta^{(\mathrm{r})}_{m,n}|\le 1.
\end{align}

It is worth noting that the model in \eqref{eq:tx_corr_model} and \eqref{eq:rx_corr_model} above separates BS array correlation from routing-induced coupling, thereby enabling a tractable characterization of the sensing structure induced by the E-FAS architecture.

\vspace{-2mm}
\subsection{Forward and Return Propagation}
Define the effective steering vector at the  transmit side as
\begin{align}
\mathbf{a}_{\mathrm{f}}^{H}(\mathbf{p}_{0})=\sum_{m=1}^{M}\rho_m \ell_{m,0}\mathbf{h}_{\mathrm{B},m}^{H},\label{eq:tx_efas_steering_corr}
\end{align}
where
\begin{align}
\ell_{m,0}=\frac{\lambda}{4\pi d_{m,0}}\exp\!\left(-j\frac{2\pi}{\lambda}d_{m,0}\right)
\end{align}
represents the free-space propagation coefficient from the $m$-th radiation point to the target \cite{liu2023integrated,hua2022mimo}. This final hop is modeled as a short-range free-space segment, while the dominant long-range routing effect is absorbed into $\rho_m$.

Similarly, the steering vector at the  receive side is given by
\begin{align}
\mathbf{a}_{\mathrm{b}}(\mathbf{p}_{0})=\sum_{m=1}^{M}\tilde{\rho}_m \bar{\ell}_{0,m}\mathbf{g}_{m,\mathrm{R}},\label{eq:rx_efas_steering_corr}
\end{align}
where $\tilde{\rho}_m$ denotes the effective return path routing coefficient and $\bar{\ell}_{0,m}$ denotes the free-space propagation coefficient from the target to the vicinity of the $m$-th radiation point. In reciprocal propagation environments, one may have
$\tilde{\rho}_m=\rho_m$ \cite{meng2024intelligent}. However, the more general formulation adopted here allows the effective forward and return routing states to differ in order to capture asymmetric routing configurations, coupling effects, and hardware non-idealities. Also, under reciprocal free-space propagation, $\bar{\ell}_{0,m}$ admits the same functional form as $\ell_{m,0}$. In general, $\tilde{\rho}_m$ need not be identical to $\rho_m$, since the transmit and receive routing paths may differ because of asymmetric coupling, routing, or hardware conditions.

\vspace{-2mm}
\subsection{Transmit Signal Model}
The BS uses a unified transmit signal for communication and sensing \cite{liu2020joint}. The transmitted baseband signal is written as
\begin{align}
\mathbf{x}(t)=\mathbf{W}\mathbf{s}(t)\in\mathbb{C}^{N_{\mathrm{t}}\times 1},
\end{align}
where $\mathbf{W}=[\mathbf{w}_{\mathrm{c}},\mathbf{w}_{\mathrm{s}}]$ denotes the precoding matrix and $\mathbf{s}(t)=[s_{\mathrm{c}}(t)\,s_{\mathrm{s}}(t)]^\top$ contains the communication and sensing probing waveforms $s_\mathrm{c}(t)$ and $s_\mathrm{s}(t)$, respectively. The signal vector $\mathbf{s}(t)$ is assumed to satisfy
\begin{align}
\mathbb{E}\!\left[\mathbf{s}(t)\mathbf{s}^H (t)\right]=\mathbf{I}_{2},
\end{align}
and the transmit power is constrained as $\mathrm{tr}\!\left(\mathbf{W}\mathbf{W}^{H}\right)\le P_{\mathrm{t}}$, where $P_{\mathrm{t}}$ denotes the total BS transmit power \cite{liu2018mu}. This model includes purely sensing transmission, purely communication transmission, and joint ISAC transmission as special cases.

\vspace{-2mm}
\subsection{Target Scattering and End-to-End Echo Model}
The target is modeled as a point scatterer with complex reflectivity $\alpha_{0}\in\mathbb{C}$, written as \cite{hua2022mimo}
\begin{align}
\alpha_{0}=\sqrt{\sigma_{\mathrm{RCS}}}\,e^{-j\varphi_{0}},
\end{align}
where $\sigma_{\mathrm{RCS}}$ denotes the effective radar cross section of the target and $\varphi_{0}$ is an unknown phase term. This point-target model isolates the propagation effects of interest without introducing additional complexity due to extended target structure. Under this model, the received echo at the BS is defined as
\begin{align}
\mathbf{y}_{\mathrm{echo}}(t)=\alpha_{0}\,\mathbf{a}_{\mathrm{b}}(\mathbf{p}_{0})\mathbf{a}_{\mathrm{f}}^{H}(\mathbf{p}_{0})\mathbf{x}(t-\tau_{0})e^{j2\pi\nu_{0}t},\label{eq:echo_signal_compact}
\end{align}
where $\nu_{0}$ denotes the Doppler shift \cite{liu2022integrated}. Equation \eqref{eq:echo_signal_compact} shows that the sensing channel is bilinear in the forward and return steering structures and is therefore fundamentally different from a conventional single manifold array sensing model.

\vspace{-2mm}
\subsection{Received Signal Model}
The practical BS receive signal contains thermal noise and, depending on the duplexing architecture, may also contain residual direct BS transmit-receive leakage. Therefore, we model the received signal as
\begin{align}
\mathbf{y}(t)=\mathbf{y}_{\mathrm{echo}}(t)+\mathbf{y}_{\mathrm{d}}(t)+\mathbf{n}(t),\label{eq:received_total}
\end{align}
where $\mathbf{n}(t)\sim\mathcal{CN}(\mathbf{0},\sigma_n^2\mathbf{I})$
denotes additive circularly symmetric complex Gaussian noise, and $\mathbf{y}_{\mathrm{d}}(t)$ collects residual direct coupling, hardware leakage, and non-ideal self-interference \cite{xiao2021full}.  To focus on the conditional sensing limits imposed by the E-FAS, the subsequent analysis assumes that the target echo is dominant, while $\mathbf{y}_{\mathrm{d}}(t)$ is either sufficiently suppressed or incorporated into the effective disturbance term.

\vspace{-2mm}
\subsection{Communication Channel}
To retain the communication functionality of the ISAC system, we define the E-FAS-assisted channel from the BS to a user located at $\mathbf{p}_{\mathrm{u}}$. The effective channel is given by
\begin{align}
\mathbf{h}_{\mathrm{u}}^{H}=\sum_{m=1}^{M}\rho_{m}\,\ell_{m,\mathrm{u}}\,\mathbf{h}_{\mathrm{B},m}^{H},\label{eq:comm_channel}
\end{align}
where
\begin{align}
\ell_{m,\mathrm{u}}\triangleq\frac{\lambda}{4\pi d_{m,\mathrm{u}}}\exp\!\left(-j\frac{2\pi}{\lambda}d_{m,\mathrm{u}}\right)
\end{align}
denotes the free-space propagation coefficient from the $m$-th radiation point to the user of interest.

The received communication signal is modeled as
\begin{align}
r_{\mathrm{u}}(t)=\mathbf{h}_{\mathrm{u}}^{H}\mathbf{x}(t)+n_{\mathrm{u}}(t),\label{eq:comm_rx_signal}
\end{align}
where $n_{\mathrm{u}}(t) \sim \mathcal{CN}(0,\sigma_{\mathrm{u}}^{2})$.

\vspace{-2mm}
\section{E-FAS-Enabled Sensing}\label{sec:channel}
This section reformulates the received echo in a parametric form suitable for estimation-theoretic analysis. The emphasis is on the spatial sensing structure by E-FAS routing. In particular, the analysis focuses on angular estimation conditioned on a given routing configuration and channel realization.

\vspace{-2mm}
\subsection{Conditional Parametric Observation Model}
The sensing analysis is carried out after delay synchronization and Doppler compensation, so that the focus remains on the angular identifiability by E-FAS. Under this model, the receive vector after matched filtering is expressed as
\begin{align}
\mathbf{y}=\alpha_0\mathbf{a}_{\mathrm{b}}(\theta_0)\mathbf{a}_{\mathrm{f}}^{H}(\theta_0)\mathbf{w}_{\mathrm{s}}+\mathbf{n},\label{eq:vector_signal_revised}
\end{align}
where $\mathbf{w}_{\mathrm{s}}$ is the sensing precoder and $\mathbf{n}\sim\mathcal{CN}(\mathbf{0},\sigma_n^2\mathbf{I})$ is the effective post-processing noise. The observation model assumes matched filtering with respect to the known sensing probing waveform, such that the communication component is treated as suppressed or separable during post-processing \cite{liu2020joint}.

By defining the conditional mean signal as
\begin{align}
\boldsymbol{\mu}(\theta_0)\triangleq\alpha_0\mathbf{a}_{\mathrm{b}}(\theta_0)\mathbf{a}_{\mathrm{f}}^{H}(\theta_0)\mathbf{w}_{\mathrm{s}},\label{eq:mean_signal_revised}
\end{align}
the observation model reduces to
\begin{equation}
\mathbf{y}=\boldsymbol{\mu}(\theta_0)+\mathbf{n}.\label{eq:obs_model_revised}
\end{equation}

The representation in \eqref{eq:obs_model_revised} is conditional in the sense that the routing coefficients and channel vectors are treated as fixed over the sensing interval, consistent with the quasi-static E-FAS routing model introduced in Section~\ref{sec:sys}.

\vspace{-2mm}
\subsection{Expanded Structure of the Mean Signal}
Using the definitions of the forward and backward steering vectors, \eqref{eq:mean_signal_revised} can be expanded as
\begin{align}
\boldsymbol{\mu}(\theta_0)=\alpha_0\sum_{m=1}^{M}\sum_{n=1}^{M}\rho_n \tilde{\rho}_m\ell_{n,0}(\theta_0)\bar{\ell}_{0,m}(\theta_0)\big(\mathbf{h}_{\mathrm{B},n}^{H}\mathbf{w}_{\mathrm{s}}\big)\mathbf{g}_{m,\mathrm{R}},\label{eq:mu_expanded_revised}
\end{align}
which illustrates that the effective sensing response is formed by the superposition of all forward-return routing pairs. This bilinear structure is a defining feature of the E-FAS sensing model and distinguishes it from conventional array sensing models, in which the sensing response is described by a single transmit steering vector and a single receive steering vector.

\vspace{-2mm}
\subsection{Ensemble Interpretation of Routing-Induced Correlation}
Although the subsequent estimation analysis is conditioned on a given realization of the routing coefficients and channel vectors, the second-order statistics introduced in Section~\ref{sec:sys} remain useful for structural interpretation. In particular, the covariance of the mean signal across routing and channel realizations can be defined as
\begin{align}
\mathbf{R}_{\mu}\triangleq\mathbb{E}\!\left[\boldsymbol{\mu}(\theta_0)\boldsymbol{\mu}^{H}(\theta_0)\right].\label{eq:R_mu_revised}
\end{align}
Substituting \eqref{eq:mu_expanded_revised} into \eqref{eq:R_mu_revised}, the covariance involves mixed terms containing the products $\rho_n \rho_{n'}^{*}$, $\mathbf{h}_{\mathrm{B},n}^{H}\mathbf{w}_{\mathrm{s}}\mathbf{w}_{\mathrm{s}}^{H}\mathbf{h}_{\mathrm{B},n'}$, and $\mathbf{g}_{m,\mathrm{R}}\mathbf{g}_{m',\mathrm{R}}^{H}$, which depend on the routing correlation coefficients $\xi_{m,n}$ and the channel correlation factors $\beta_{m,n}^{(\mathrm{t})}$ and $\beta_{m,n}^{(\mathrm{r})}$. Therefore, the second-order structure of the E-FAS sensing response is jointly shaped by routing-induced coupling and BS-side spatial correlation. This statistical interpretation is used only to explain structural trends and is not substituted directly into the conditional Fisher information analysis. It is worth noting that the ensemble correlation model introduced above is used only to characterize the second-order structure of routing-induced coupling. The Fisher information and CRB analysis developed later remain fully conditional on a fixed realization of the routing state and channel vectors.

\vspace{-2mm}
\section{Fisher Information and CRB Analysis}\label{sec:crb}
This section characterizes the fundamental limit of angular estimation for the conditional E-FAS sensing model. Throughout the analysis, the routing state and channel realizations are treated as deterministic, consistent with the quasi-static assumption introduced in Section~\ref{sec:sys}.

\vspace{-2mm}
\subsection{Likelihood Function}
From \eqref{eq:obs_model_revised}, the observation vector $\mathbf{y}$ follows a circularly symmetric complex Gaussian distribution with mean $\boldsymbol{\mu}(\theta_0)$ and covariance $\sigma_n^2\mathbf{I}$ \cite{trees2002optimum,stoica1989music}. Therefore, conditioned on $\theta_0$, the likelihood function is given by
\begin{align}
p(\mathbf{y};\theta_0)=\frac{1}{(\pi\sigma_n^2)^{N_{\mathrm r}}}\exp\!\left(-\frac{1}{\sigma_n^2}\left\|\mathbf{y}-\boldsymbol{\mu}(\theta_0)\right\|^2\right),
\end{align}
and the corresponding log-likelihood is written as
\begin{align}
\mathcal{L}(\theta_0)=-\frac{1}{\sigma_n^2}\left\|\mathbf{y}-\boldsymbol{\mu}(\theta_0)\right\|^2+\mathrm{const}.\label{eq:loglikelihood_final}
\end{align}

\vspace{-2mm}
\subsection{Fisher Information}
Since the covariance matrix is independent of $\theta_0$, the Fisher information is fully determined by the sensitivity of the conditional mean signal. Using the standard Fisher information expression for complex Gaussian observations \cite{stoica1989music}, we obtain
\begin{align}
J_{\theta\theta}=\frac{2}{\sigma_n^2}\Re\!\left\{\left(\frac{\partial \boldsymbol{\mu}(\theta_0)}{\partial \theta_0}\right)^{H}\left(\frac{\partial \boldsymbol{\mu}(\theta_0)}{\partial \theta_0}\right)\right\}.\label{eq:FIM_final}
\end{align}
Equivalently,
\begin{align}
J_{\theta\theta}=\frac{2}{\sigma_n^2}\left\|\frac{\partial \boldsymbol{\mu}(\theta_0)}{\partial \theta_0}\right\|^2.
\end{align}
From \eqref{eq:mean_signal_revised}, the derivative of the mean signal is given by
\begin{align}
\frac{\partial \boldsymbol{\mu}(\theta_0)}{\partial \theta_0}=\alpha_0\left(\frac{\partial \mathbf{a}_{\mathrm b}(\theta_0)}{\partial \theta_0} \mathbf{a}_{\mathrm f}^{H}(\theta_0) +\mathbf{a}_{\mathrm b}(\theta_0)\frac{\partial \mathbf{a}_{\mathrm f}^{H}(\theta_0)}{\partial \theta_0}\right)\mathbf{w}_{\mathrm s}.\label{eq:dmu_final}
\end{align}
The derivatives of the steering vectors are given by
\begin{align}
\frac{\partial \mathbf{a}_{\mathrm f}^{H}(\theta_0)}{\partial \theta_0}&=\sum_{n=1}^{M}\rho_n\frac{\partial \ell_{n,0}(\theta_0)}{\partial \theta_0}\mathbf{h}_{\mathrm B,n}^{H},\\
\frac{\partial \mathbf{a}_{\mathrm b}(\theta_0)}{\partial \theta_0}&=\sum_{m=1}^{M}\tilde{\rho}_m\frac{\partial \bar{\ell}_{0,m}(\theta_0)}{\partial \theta_0}\mathbf{g}_{m,\mathrm R}.
\end{align}
Substituting these expressions into \eqref{eq:FIM_final} shows that the Fisher information is governed by the coherent interaction of all forward-return routing pairs. In particular, both the routing coefficients and the geometric sensitivity terms $\partial \ell/\partial\theta_0$ jointly determine the achievable angular estimation accuracy.

\vspace{-2mm}
\subsection{CRB}
The CRB provides the fundamental lower bound on the variance of any unbiased estimator of the angular parameter under the conditional observation model in \eqref{eq:obs_model_revised}. Since the unknown parameter is scalar, the angular CRB is obtained directly from
the inverse of the Fisher information as
\begin{align}
\mathrm{CRB}_{\theta}\triangleq\left[J_{\theta\theta}\right]^{-1}=\frac{1}{J_{\theta\theta}},\label{eq:CRB_final}
\end{align}
provided that $J_{\theta\theta}>0$. Thus, for any unbiased estimator $\hat{\theta}_0$,
\begin{align}
\mathrm{Var}(\hat{\theta}_0)\ge\mathrm{CRB}_{\theta}.
\end{align}
Using \eqref{eq:FIM_final}, the CRB can be equivalently written as
\begin{align}
\mathrm{CRB}_{\theta}=\frac{\sigma_n^2}{2\left\|\frac{\partial \boldsymbol{\mu}(\theta_0)}{\partial \theta_0}\right\|^2}.\label{eq:CRB_derivative_form}
\end{align}
This expression shows that angular estimation accuracy is governed by the sensitivity of the conditional mean response to perturbations in $\theta_0$. In the considered E-FAS sensing model, this sensitivity is not determined solely by the received signal power. Rather, it depends jointly on the effective routing coefficients, the forward and return channel realizations, and the angular variation of the final free-space propagation coefficients associated with the distributed radiation points. Therefore, a small CRB requires not only large effective routing gain, but also sufficient angular diversity in the composite E-FAS sensing response. If the routed paths provide strong received power but induce nearly identical angular variations, the norm $\left\|\partial \boldsymbol{\mu}(\theta_0)/\partial\theta_0\right\|$ may remain limited, resulting in poor angular resolution. Conversely, spatially diverse radiation points can increase this derivative norm by creating distinguishable angle-dependent phase and amplitude variations across the forward and return paths. This observation motivates the scaling and structural analysis developed in the next section.

\vspace{-2mm}
\subsection{A Rank-Deficient Sensing Regime}
The following proposition characterizes a sufficient condition under which the E-FAS sensing channel collapses to a rank-one spatial structure, thereby limiting the number of independent sensing modes.

\begin{proposition}\label{prop1}
Suppose that the routing coefficients satisfy
\begin{align}
\rho_m=c_{\rho}\sqrt{\gamma_m},~\tilde{\rho}_m=c_{\tilde{\rho}}\sqrt{\tilde{\gamma}_m},~\forall m,\label{eq:routing_alignment}
\end{align}
for some complex scalars $c_{\rho}$ and $c_{\tilde{\rho}}$. Furthermore, assume that there exist vectors $\mathbf{h}_{\mathrm B}\in\mathbb{C}^{N_{\mathrm t}\times 1}$ and $\mathbf{g}_{\mathrm R}\in\mathbb{C}^{N_{\mathrm r}\times 1}$ such that
\begin{align}
\mathbf{h}_{\mathrm B,m}=a_m\mathbf{h}_{\mathrm B},~\mathbf{g}_{m,\mathrm R}=b_m\mathbf{g}_{\mathrm R},~\forall m,\label{eq:channel_alignment}
\end{align}
where $\{a_m\}$ and $\{b_m\}$ are complex scalars. Then the conditional mean signal $\boldsymbol{\mu}(\theta_0)$ admits the factorization
\begin{align}
\boldsymbol{\mu}(\theta_0)=\psi(\theta_0)\,\mathbf{g}_{\mathrm R},
\end{align}
for some scalar function $\psi(\theta_0)$. Consequently,
\begin{align}
\mathrm{rank}\!\left(\boldsymbol{\mu}(\theta_0)\boldsymbol{\mu}^{H}(\theta_0)\right)=1,
\end{align}
and the derivative $\partial\boldsymbol{\mu}(\theta_0)/\partial\theta_0$ remains confined to the same one-dimensional subspace. Hence, the Fisher information is generated by a single effective spatial sensing mode.
\end{proposition}

\begin{proof}	
Substituting \eqref{eq:routing_alignment} and \eqref{eq:channel_alignment} into the transmit steering vector yields
\begin{align}
\mathbf{a}_{\mathrm f}^{H}(\theta_0)&=\sum_{m=1}^{M}c_{\rho}\sqrt{\gamma_m}\,\ell_{m,0}(\theta_0)a_m^{*}\mathbf{h}_{\mathrm B}^{H}=\chi_{\mathrm f}(\theta_0)\mathbf{h}_{\mathrm B}^{H},
\end{align}
where
\begin{align}
\chi_{\mathrm f}(\theta_0)\triangleq\sum_{m=1}^{M}c_{\rho}\sqrt{\gamma_m}\,\ell_{m,0}(\theta_0)a_m^{*} 
\end{align}
is a scalar-valued function.

Similarly, the receive steering vector becomes
\begin{align}
\mathbf{a}_{\mathrm b}(\theta_0)&=\sum_{m=1}^{M}c_{\tilde{\rho}}\sqrt{\tilde{\gamma}_m}\,\bar{\ell}_{0,m}(\theta_0)b_m\mathbf{g}_{\mathrm R}=\chi_{\mathrm b}(\theta_0)\mathbf{g}_{\mathrm R},
\end{align}
where $\chi_{\mathrm b}(\theta_0)$ is also scalar-valued.

Substituting the above expressions into the conditional mean signal model yields
\begin{align}
\boldsymbol{\mu}(\theta_0)&=\alpha_0\chi_{\mathrm b}(\theta_0)\chi_{\mathrm f}(\theta_0)\mathbf{g}_{\mathrm R}\mathbf{h}_{\mathrm B}^{H}\mathbf{w}_{\mathrm s}=\psi(\theta_0)\,\mathbf{g}_{\mathrm R},
\end{align}
where
\begin{equation}
\psi(\theta_0)=\alpha_0\chi_{\mathrm b}(\theta_0)\chi_{\mathrm f}(\theta_0)\mathbf{h}_{\mathrm B}^{H}\mathbf{w}_{\mathrm s}
\end{equation}
is scalar-valued. Therefore, both $\boldsymbol{\mu}(\theta_0)$ and its derivative with respect to $\theta_0$ belong to the one-dimensional subspace spanned by $\mathbf{g}_{\mathrm R}$. Consequently, the corresponding signal covariance matrix is rank one, and the Fisher information is generated through a single effective sensing mode.
\end{proof}

Proposition~\ref{prop1} shows that strong routing alignment combined with aligned channel responses can eliminate spatial diversity in the effective E-FAS sensing channel. In this regime, increasing the number of radiation points does not increase the number of independent sensing dimensions, since all routed contributions collapse onto the same spatial subspace.

\vspace{-2mm}
\section{Insights and Fundamental Limits}\label{sec:limits}
In this section, we interpret the Fisher information derived in the previous section in order to characterize the roles of SW routing gain, routing-induced correlation, and spatial geometry in the proposed E-FAS sensing architecture. Unless otherwise specified, all quantities are conditioned on a fixed realization of the routing state and the associated channel vectors.

\vspace{-2mm}
\subsection{Preliminaries}
From \eqref{eq:mean_signal_revised}, the conditional mean signal is given by
\begin{align}
\boldsymbol{\mu}(\theta_0)=\alpha_0\mathbf{a}_{\mathrm{b}}(\theta_0)\mathbf{a}_{\mathrm{f}}^{H}(\theta_0)\mathbf{w}_{\mathrm{s}}.
\end{align}
To expose the contributions of the forward and return propagation structures, we define the effective forward gain as
\begin{align}
G_{\mathrm f}\triangleq\left|\mathbf{a}_{\mathrm f}^{H}(\theta_0)\mathbf{w}_{\mathrm s}\right|^2,
\end{align}
and the effective return gain as
\begin{align}
G_{\mathrm b}\triangleq\left\|\mathbf{a}_{\mathrm b}(\theta_0)\right\|^2.
\end{align}
Under the conditional observation model, the effective sensing signal-to-noise ratio (SNR) can be expressed as
\begin{equation}
\mathrm{SNR}_{\mathrm s}=\frac{|\alpha_0|^2G_{\mathrm f}G_{\mathrm b}}{\sigma_n^2}.\label{eq:SNR_final}
\end{equation}

Equation \eqref{eq:SNR_final} characterizes the overall amplification of the sensing response induced by the E-FAS routing structure. But the Fisher information depends not only on the received signal strength, but also on the angular sensitivity of the composite sensing manifold. Therefore, a larger sensing SNR alone does not necessarily imply higher angular estimation accuracy. In particular, if different routing paths generate nearly identical angular responses, the derivative $\partial\boldsymbol{\mu}(\theta_0)/\partial\theta_0$ may remain small even when the received signal power is large. The following analysis formalizes how routing diversity and spatial structure jointly determine the achievable sensing resolution.

\vspace{-2mm}
\subsection{First-Order Angular Identifiability}
Before deriving scaling laws, it is useful to state the regularity condition under which the angular parameter contributes nonzero Fisher information. As the observation covariance is independent of $\theta_0$, the angular information is determined by the sensitivity of the conditional mean vector $\boldsymbol{\mu}(\theta_0)$.

\begin{lemma}[First-Order Angular Identifiability]\label{lemma1}
Under the conditional Gaussian observation model in \eqref{eq:obs_model_revised}, the angular Fisher information is strictly positive if and only if
\begin{align}
\frac{\partial \boldsymbol{\mu}(\theta_0)}{\partial \theta_0}\neq \mathbf{0}.\label{eq:first_order_identifiability}
\end{align}
Equivalently, if \eqref{eq:first_order_identifiability} holds, then $\theta_0$ is locally identifiable in the first-order Fisher-information sense;	otherwise, the angular CRB is unbounded.
\end{lemma}

\begin{proof}
From \eqref{eq:FIM_final}, the angular Fisher information is given by
\begin{align}
J_{\theta\theta}=\frac{2}{\sigma_n^2}\left\|\frac{\partial \boldsymbol{\mu}(\theta_0)}{\partial \theta_0}\right\|^2.
\end{align}
Since $\sigma_n^2>0$, we have $J_{\theta\theta}>0$ if and only if $\partial \boldsymbol{\mu}(\theta_0)/\partial\theta_0\neq\mathbf{0}$. If this derivative is zero, then $J_{\theta\theta}=0$, and the scalar CRB, $\mathrm{CRB}_{\theta}=1/J_{\theta\theta}$, becomes unbounded.	Conversely, if the derivative is nonzero, then $J_{\theta\theta}>0$, and the observation distribution changes to 	first order with respect to perturbations of $\theta_0$. Thus, the proof is accomplished.
\end{proof}

It is worth recognizing that Lemma \ref{lemma1} shows that angular resolvability in the proposed E-FAS sensing model requires the composite mean response to vary with the target angle. As such, received signal power alone is not sufficient; the routed forward and return paths must also induce a nonzero angular perturbation in the effective sensing manifold.

\vspace{-2mm}
\subsection{Scaling with Routing Gain}
We next quantify how the effective routing gain enters the angular Fisher information. To do so, we define the normalized angular sensitivities as
\begin{align}
\delta_{\mathrm f}(\theta_0)&\triangleq\frac{\left|\frac{\partial \mathbf{a}_{\mathrm f}^{H}(\theta_0)}{\partial \theta_0}\mathbf{w}_{\mathrm s}\right|}{\left|\mathbf{a}_{\mathrm f}^{H}(\theta_0)\mathbf{w}_{\mathrm s}\right|},\\
\delta_{\mathrm b}(\theta_0)&\triangleq\frac{\left\|\frac{\partial \mathbf{a}_{\mathrm b}(\theta_0)}{\partial \theta_0}\right\|}{\left\|\mathbf{a}_{\mathrm b}(\theta_0)\right\|},
\end{align}
whenever the denominators are nonzero.

\begin{proposition}[Routing-Gain Fisher Information Bound]
Assume that $G_{\mathrm f}>0$ and $G_{\mathrm b}>0$. Then the angular Fisher information satisfies
\begin{align}
J_{\theta\theta}\le\frac{2|\alpha_0|^2}{\sigma_n^2}G_{\mathrm f}G_{\mathrm b}\left(\delta_{\mathrm f}(\theta_0)+\delta_{\mathrm b}(\theta_0)\right)^2.\label{eq:J_gain_upper_bound}
\end{align}
Also, assuming that there exists a constant $\zeta > 0$ such that
\begin{align}
\left\|\frac{\partial \boldsymbol{\mu}(\theta_0)}{\partial\theta_0}\right\|\ge|\alpha_0|\sqrt{G_{\mathrm f}G_{\mathrm b}}\,\zeta,\label{eq:no_cancel_condition}
\end{align}
then we have
\begin{align}
J_{\theta\theta}\ge\frac{2|\alpha_0|^2}{\sigma_n^2}G_{\mathrm f}G_{\mathrm b}\zeta^2.\label{eq:J_gain_lower_bound}
\end{align}
\end{proposition}

\begin{proof}
From \eqref{eq:dmu_final} and by using the triangle inequality, we have
\begin{align}
\left\|\frac{\partial \boldsymbol{\mu}(\theta_0)}{\partial\theta_0}\right\|&\le|\alpha_0|\left(\left\|\frac{\partial \mathbf{a}_{\mathrm b}(\theta_0)}{\partial\theta_0}\right\|\left|\mathbf{a}_{\mathrm f}^{H}(\theta_0)\mathbf{w}_{\mathrm s}\right|\right.\nonumber\\
&\quad\left.+\left\|\mathbf{a}_{\mathrm b}(\theta_0)\right\|\left|\frac{\partial \mathbf{a}_{\mathrm f}^{H}(\theta_0)}{\partial\theta_0}\mathbf{w}_{\mathrm s}\right|\right)\nonumber\\
&=|\alpha_0|\sqrt{G_{\mathrm f}G_{\mathrm b}}\left(\delta_{\mathrm f}(\theta_0)+\delta_{\mathrm b}(\theta_0)\right).
\end{align}
Squaring both sides and substituting into \eqref{eq:FIM_final} yields \eqref{eq:J_gain_upper_bound}. The lower bound follows directly from \eqref{eq:no_cancel_condition} and \eqref{eq:FIM_final}.
\end{proof}

Equation~\eqref{eq:J_gain_upper_bound} shows that the product $G_{\mathrm f}G_{\mathrm b}$ controls the power scaling of the angular Fisher information only through the angular sensitivities of the composite forward and return steering structures. Thus, routing gain is necessary for improving the sensing SNR, but it is not sufficient for high angular resolution unless the routed paths also produce distinguishable angle-dependent sensing responses.

\vspace{-2mm}
\subsection{Role of Spatial Geometry}
Here, we characterize how the spatial distribution of the effective radiation points affects angular sensing performance. In the proposed E-FAS architecture, the angular Fisher information depends on the variation of the propagation coefficients with respect to the target angle, which is determined fundamentally by the underlying geometry of the routed paths.

\begin{proposition}[Geometry-Dependent Information Accumulation]\label{prop3}
Assume that the routing-induced correlation coefficients satisfy
\begin{align}
|\beta_{m,n}^{(\mathrm t)}|\le \epsilon,~|\beta_{m,n}^{(\mathrm r)}|\le \epsilon,~m\neq n,
\end{align}
for some sufficiently small $\epsilon>0$. Also, assume $c_1,c_2,c_3,c_4>0$ such that
\begin{align}
c_1\le|\rho_m|,~|\tilde{\rho}_m|\le c_2,
\end{align}
and
\begin{align}
c_3\le \|\mathbf h_{\mathrm B,m}\|, \|\mathbf g_{m,\mathrm R}\|\le c_4,
\end{align}
for all active routing paths. Then there exist positive constants $C_{\mathrm L}$ and $C_{\mathrm U}$ independent of $M$ such that
\begin{align}
C_{\mathrm L}\,\Psi(\theta_0)\le J_{\theta\theta}\le C_{\mathrm U}\,\Psi(\theta_0),\label{eq:geometry_bound_final}
\end{align}
where
\begin{align}
\Psi(\theta_0)\triangleq\sum_{m=1}^{M}\left|\frac{\partial \ell_{m,0}(\theta_0)}{\partial\theta_0}\right|^2+\sum_{m=1}^{M}\left|\frac{\partial \bar{\ell}_{0,m}(\theta_0)}{\partial\theta_0}\right|^2.
\end{align}
\end{proposition}

\begin{proof}
Recalling from \eqref{eq:dmu_final}, the derivative of the conditional mean vector is given by a weighted superposition of forward and return steering vector derivatives. Substituting \eqref{eq:dmu_final} into \eqref{eq:FIM_final} then yields a quadratic form containing both within-path terms and inter-path cross terms.

The within-path terms are proportional to
\begin{align}
\left|\frac{\partial \ell_{m,0}(\theta_0)}{\partial\theta_0}\right|^2~\text{and}~\left|\frac{\partial \bar{\ell}_{0,m}(\theta_0)}{\partial\theta_0}\right|^2,
\end{align}
modulated by routing gains and channel norms. Since these quantities are uniformly bounded above and below, they contribute only bounded multiplicative factors.

On the other hand, the inter-path terms are weighted by the correlation coefficients $\beta_{m,n}^{(\mathrm t)}$ and $\beta_{m,n}^{(\mathrm r)}$. Under the assumption $|\beta_{m,n}^{(\mathrm t)}|, |\beta_{m,n}^{(\mathrm r)}|\le\epsilon$ for $m\neq n$, the aggregate cross term contribution remains bounded by a constant factor perturbation of the dominant within-path terms. Consequently, the Fisher information is bounded above and below by positive constant multiples of $\Psi(\theta_0)$, which proves \eqref{eq:geometry_bound_final}.
\end{proof}

Interestingly, Proposition~\ref{prop3} shows that angular sensing performance improves when the active radiation points generate sufficiently distinct angle-dependent propagation responses. As a result, from a sensing perspective, the spatial placement of the effective radiation points is important not only for signal coverage, but also  for creating geometrically diverse angular perturbations across the routed paths.

\vspace{-2mm}
\subsection{Gain-Diversity Trade-off}
The preceding results reveal an inherent trade-off between routing gain and sensing diversity in the considered E-FAS architecture. In particular, concentrating the routing resources over a small number of high-gain paths increases the effective sensing SNR, whereas distributing the routing energy across multiple spatially distinct paths enhances the diversity of the angle-dependent sensing responses. To illustrate this behavior, consider a fixed routing power budget
\begin{align}
\sum_{m=1}^{M}\gamma_m=\Gamma,
\end{align}
where $\gamma_m=\mathbb{E}[|\rho_m|^2]$ represents the average routing power associated with the $m$-th radiation path. Under this constraint, allocating most of the routing power to a small subset of paths increases the dominant steering gain and therefore improves the effective sensing SNR defined in \eqref{eq:SNR_final}. However, strong concentration of the routing coefficients also increases the likelihood that the resulting forward and return steering structures become highly aligned, thereby reducing the effective dimensionality of the sensing response. Conversely, distributing the routing power over multiple geometrically distinct radiation points generally reduces the gain contributed by any individual path, but increases the spatial variability of the composite steering structure. As shown in the previous subsection, this geometrical diversity improves the accumulation of angular information through distinct angle-dependent perturbations across the routed paths. Therefore, E-FAS-enabled sensing exhibits a fundamental gain-diversity trade-off: routing strategies that maximize coherent signal power do not necessarily maximize angular identifiability, while highly diverse routing configurations may sacrifice coherent amplification in exchange for improved spatial resolution capability.

\vspace{-2mm}
\subsection{Extension to Joint Delay Angle Estimation}
Although the main analysis of this paper focuses on angular estimation after delay synchronization, the proposed framework can be extended to joint delay angle estimation in a straightforward manner. In particular, for a wideband sensing waveform, the conditional mean signal becomes jointly parameterized by the delay $\tau_0$ and the angle $\theta_0$. Let $\boldsymbol{\eta}=[\tau_0,\theta_0]^T$ denote the unknown parameter vector. The corresponding conditional mean observation can be expressed as
\begin{align}
\boldsymbol{\mu}(\tau_0,\theta_0)=\alpha_0\mathbf{a}_{\mathrm b}(\theta_0)\mathbf{a}_{\mathrm f}^{H}(\theta_0)\mathbf{w}_{\mathrm s}\,s_{\mathrm s}(t-\tau_0),
\end{align}
where $s_{\mathrm s}(t)$ denotes the sensing waveform.

Under the conditional Gaussian observation model, the associated FIM is given by
\begin{align}
\mathbf{J}(\boldsymbol{\eta})=
\begin{bmatrix}
J_{\tau\tau} & J_{\tau\theta}\\
J_{\theta\tau} & J_{\theta\theta}
\end{bmatrix},
\end{align}
whose entries satisfy
\begin{align}
[\mathbf{J}(\boldsymbol{\eta})]_{i,j}=\frac{2}{\sigma_n^2}\Re\!\left\{\left(\frac{\partial\boldsymbol{\mu}}{\partial\eta_i}\right)^H\left(\frac{\partial\boldsymbol{\mu}}{\partial\eta_j}\right)\right\},
\end{align}
where $\eta_1=\tau_0$ and $\eta_2=\theta_0$.

The resulting structure shows that delay estimation is primarily governed by the temporal characteristics of the sensing waveform, whereas angular estimation is mainly determined  by the spatial variation of the composite E-FAS steering structure. The off-diagonal terms $J_{\tau\theta}$ and $J_{\theta\tau}$ quantify the coupling between temporal and spatial estimation mechanisms.

\vspace{-4mm}
\section{Numerical Results}\label{sec:num}
In this section, we validate the analytical framework developed in the previous sections and investigate the sensing performance of E-FAS-enabled ISAC systems. In particular, we examine the impact of routing gain, routing-induced correlation, and spatial distribution of effective radiation points on the CRB. Unless otherwise specified, the carrier frequency is set to $f_c=28~\mathrm{GHz}$, corresponding to a wavelength $\lambda \approx 10.7~\mathrm{mm}$. The BS employs $N_{\mathrm t}=N_{\mathrm r}=8$ transmit and receive antennas with half-wavelength antenna spacing. The target angle is fixed at $\theta_0=25^\circ$, while the target range is set to
$r_0=16~\mathrm{m}$. The E-FAS radiation surface is located at $x_{\mathrm{surf}}=6~\mathrm{m}$, and the default E-FAS aperture size is $D=3~\mathrm{m}$. Unless otherwise stated, free-space propagation is assumed for all wireless links. The BS--E-FAS, E-FAS--target, target--E-FAS, and E-FAS--BS channels follow the spherical-wave propagation model introduced in Section \ref{sec:sys}, with free-space path-loss exponent equal to $2$. Furthermore, the attenuation of the guided routing medium is absorbed into the routing coefficients $\{\rho_m,\tilde{\rho}_m\}$ and the routing power budget constraint $\sum_m |\rho_m|^2 \leq \Gamma$. Furthermore, the number of effective radiation points is chosen as $M=16$, and the routing-power budget is normalized to $\Gamma=1$~(0~dB). The sensing precoder is selected as
$
\mathbf w_{\mathrm s}
=
\frac{1}{\sqrt{N_t}}
\mathbf 1_{N_t},
$
where $\mathbf 1_{N_t}$ denotes the all-ones vector of length $N_t$. 
Besides, all results are averaged over independently generated routing and geometry realizations whenever random placement is considered. To validate the derived CRB, Monte Carlo simulations are performed using the conditional observation model  with $5000$ independent noise realizations. In particular, local maximum likelihood angular estimation is implemented under additive complex Gaussian noise, and the resulting root mean square error (RMSE) is compared with the theoretical $\sqrt{\mathrm{CRB}_{\theta}}$. 

\begin{figure}[t]
\vspace{-2cm}	\centering
	\includegraphics[width=1\columnwidth]{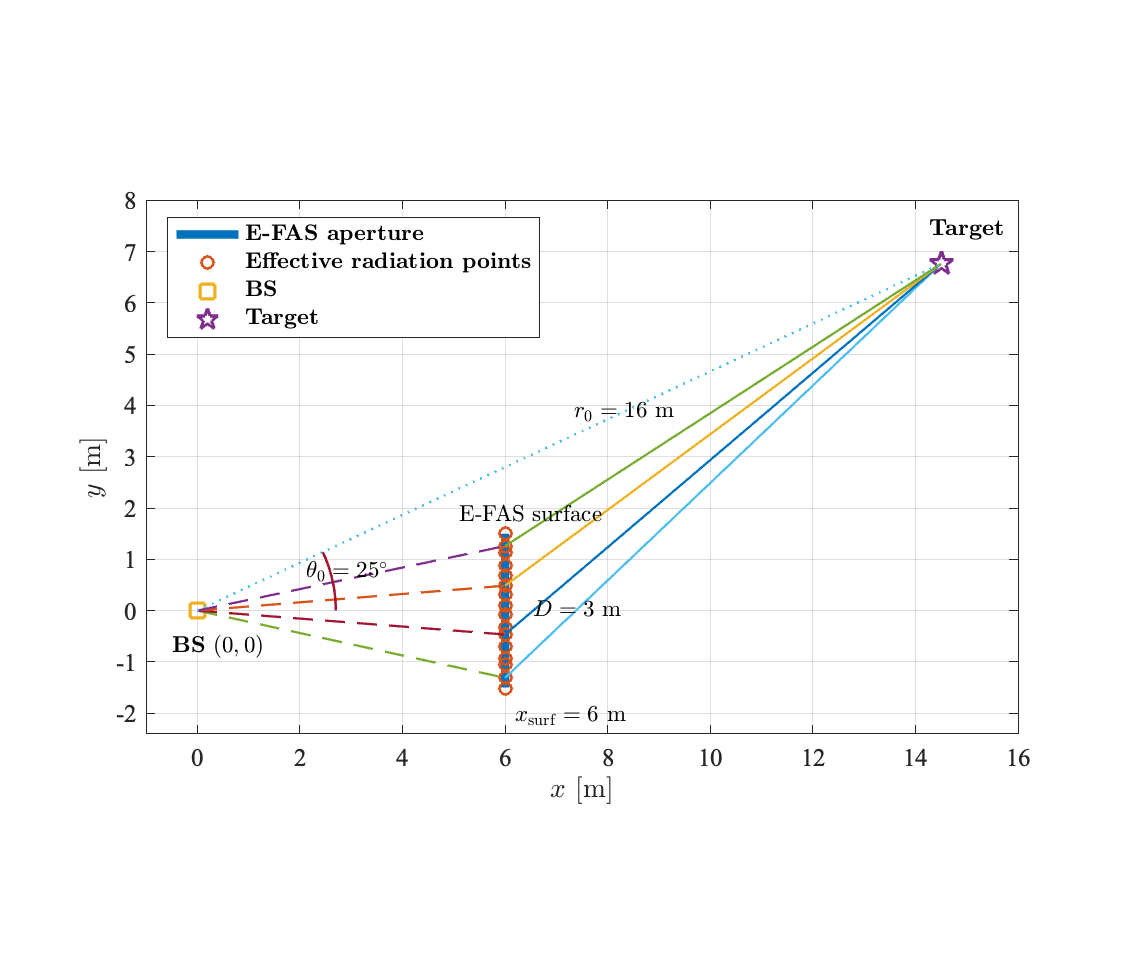}
\vspace{-1.5cm}	\caption{Representative simulation geometry used in the numerical evaluations.}
\label{fig:ge-model}
	\vspace{-2mm}
\end{figure}

Fig.~\ref{fig:ge-model} illustrates a representative simulation geometry used throughout the numerical evaluations. Unless otherwise stated, the BS is located at the origin, while the E-FAS radiation surface is positioned at $x_{\mathrm{surf}}=6~\mathrm{m}$ with aperture size $D=3~\mathrm{m}$. The target is located at $(r_0,\theta_0)=(16~\mathrm{m},25^\circ)$ and the distributed radiation points are generated over the E-FAS aperture according to the considered routing realization.

\begin{figure}[t]
\centering
\includegraphics[width=0.85\columnwidth]{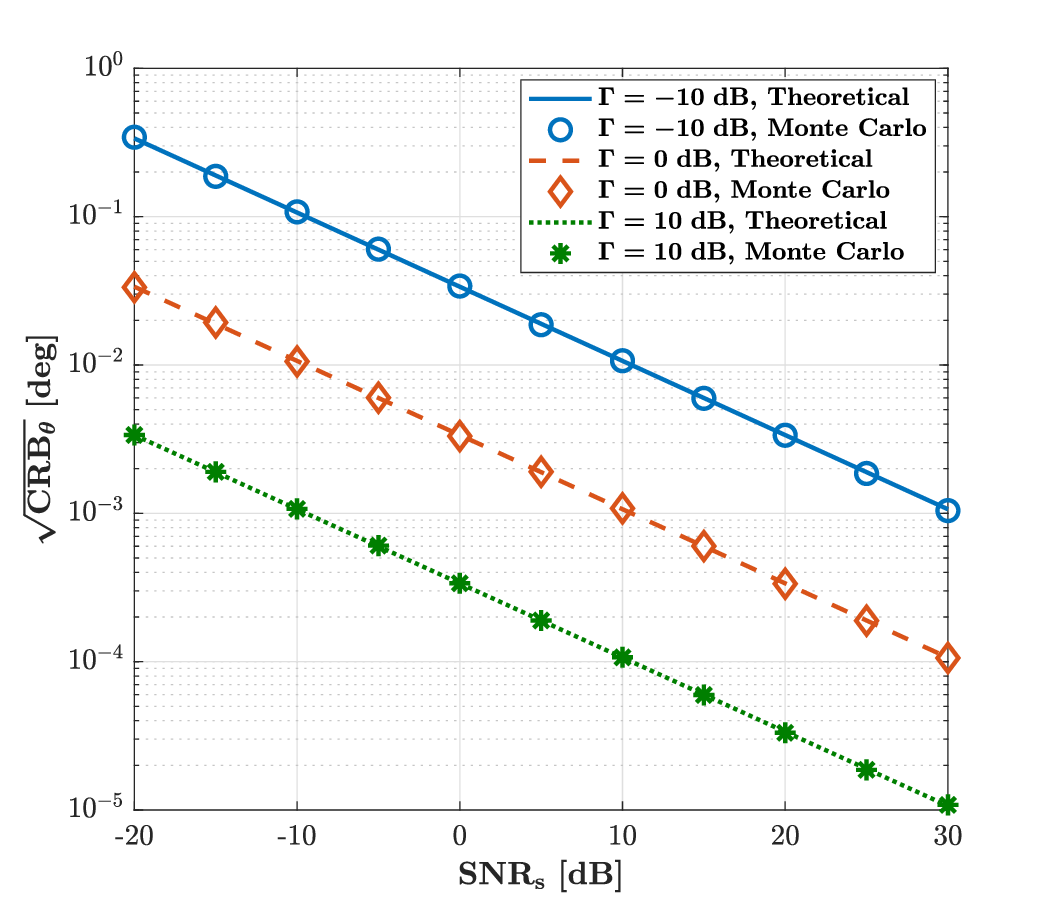}
\vspace{-2mm}\caption{Angular CRB validation versus sensing SNR $\mathrm{SNR}_{\mathrm s}$ for different routing power budgets $\Gamma$.}\label{fig:crb-snr}
\vspace{-6mm}
\end{figure}

Fig.~\ref{fig:crb-snr} illustrates the CRB validation and the impact of the routing power budget $\Gamma$ on the angular sensing performance for different sensing SNR values. The figure compares the theoretical $\sqrt{\mathrm{CRB}_{\theta}}$ obtained from the proposed E-FAS  Fisher information formulation with the Monte Carlo RMSE achieved by local maximum likelihood estimation under the conditional bidirectional sensing model. It is observed that the Monte Carlo results closely match the theoretical CRB curves over the entire SNR range and for all considered routing power budgets. This agreement confirms the correctness of the derived analytical framework and verifies that the proposed estimator operates in the efficient estimation regime.  

Moreover, the figure clearly shows that increasing the routing power budget substantially improves the angular estimation accuracy. For example, the case with $\Gamma=10$~dB achieves significantly lower angular error compared with $\Gamma=0$~dB and $\Gamma=-10$~dB over the entire SNR region. Specifically, the proposed Fisher information expression scales proportionally with the effective coherent routing gain induced by the E-FAS surface. A larger routing power budget enables stronger coherent accumulation of the routed sensing field, which increases the derivative energy of the received signal with respect to the target angle. Consequently, the Fisher information increases, leading to a lower CRB. In addition, all curves exhibit the expected monotonic decay with increasing sensing SNR. This behavior follows directly from the inverse dependence of the CRB on the noise variance. As the sensing SNR increases, the received observation becomes less noise limited, allowing more accurate extraction of angular perturbations from the received waveform. Thus, the estimation variance decreases almost linearly on the logarithmic scale. As a result, the coherent routing amplification provided by E-FAS architectures can substantially enhance sensing accuracy and reduce the fundamental estimation bound in ISAC systems.

\begin{figure}[t]
\centering
\includegraphics[width=0.85\columnwidth]{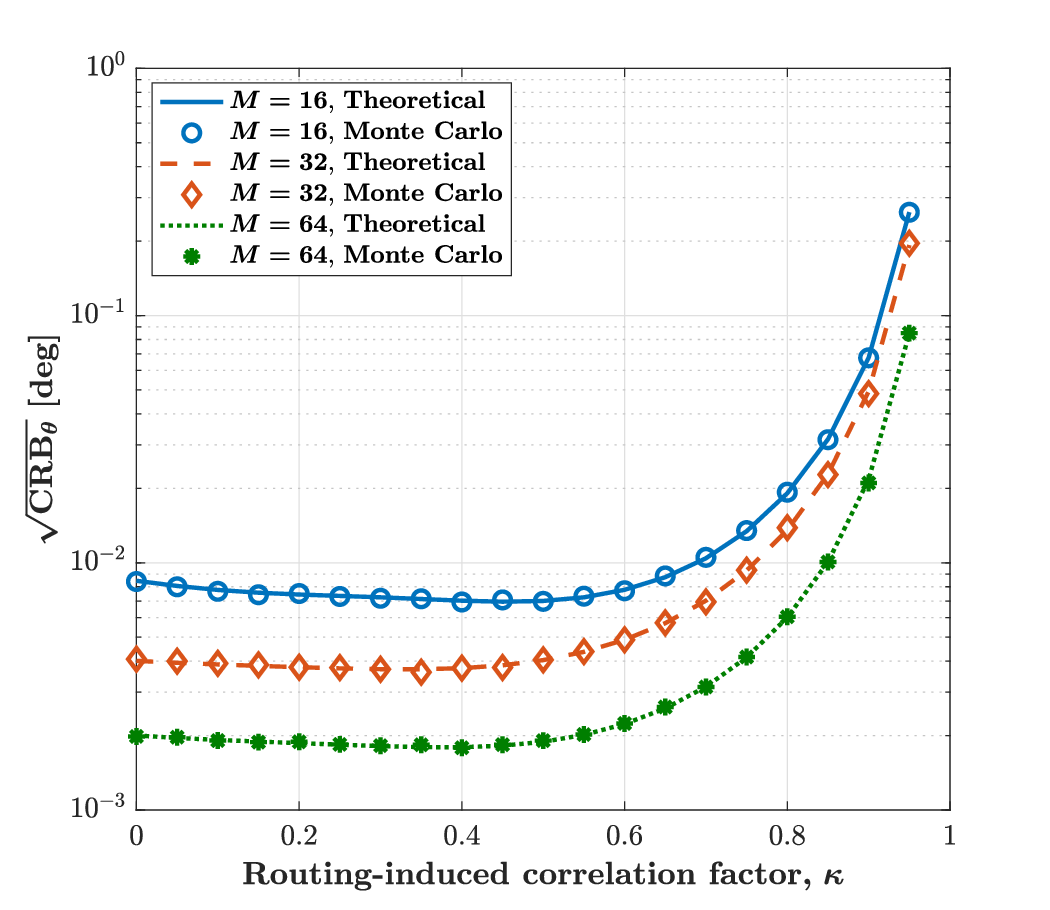}
\vspace{-2mm}\caption{Angular CRB versus routing-induced correlation factor $\kappa$ for different numbers of effective radiation points $M$.}\label{fig:crb-kappa}
\vspace{-6mm}
\end{figure}

Fig.~\ref{fig:crb-kappa} illustrates the impact of the routing-induced correlation factor $\kappa$ on the angular sensing performance for different numbers of effective radiation points $M$. We observe that the angular estimation error initially decreases slightly for small correlation levels and then increases rapidly as the routing-induced correlation becomes strong. This behavior is directly related to the effective sensing diversity created by the distributed radiation points. For weak correlation regimes, a moderate increase in $\kappa$ improves the coherent aggregation of the routed sensing field, which slightly enhances the received signal strength and reduces the CRB. However, as $\kappa$ becomes large, the effective radiation modes gradually collapse into highly correlated sensing components. As such, the spatial diversity of E-FAS decreases significantly, leading to a reduction in the effective angular sensitivity of the received waveform. This effect becomes severe when $\kappa$ approaches unity, where the distributed radiation structure behaves similarly to a highly correlated compact aperture. In this regime, the derivative energy of the conditional received signal with respect to the target angle decreases substantially, which directly reduces the Fisher information and causes a rapid increase in the CRB. The results further show that increasing the number of effective radiation points consistently improves the sensing accuracy over the entire correlation region. Larger values of $M$ provide richer spatial sampling of the propagated sensing field and generate higher effective sensing diversity. Consequently, the proposed E-FAS architecture becomes more resilient against routing-induced correlation effects, resulting in lower estimation error compared to the cases with smaller values of $M$.

\begin{figure}[t]
\centering
\includegraphics[width=0.85\columnwidth]{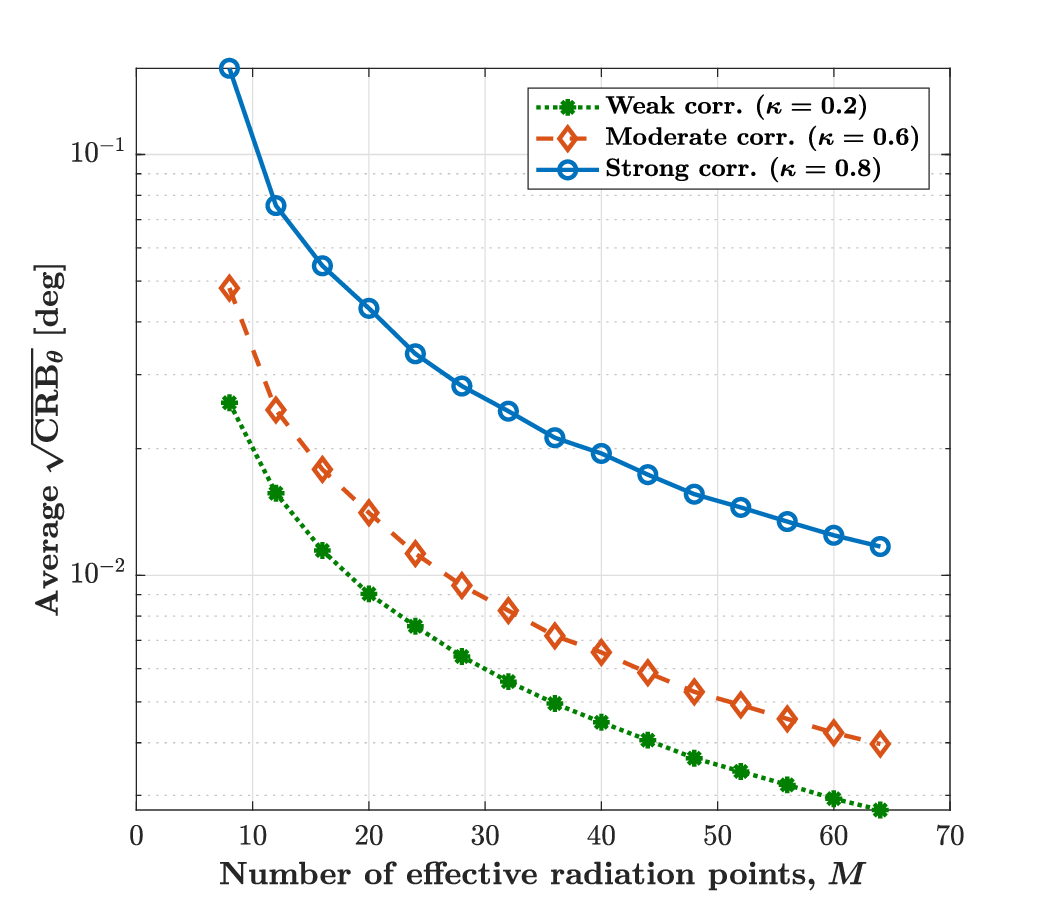}
\vspace{-2mm}\caption{Average CRB versus the number of effective radiation points $M$ for different routing-induced correlation levels.}\label{fig:crb-m}
\vspace{-4mm}
\end{figure}

Fig.~\ref{fig:crb-m} studies the impact of the number of effective radiation points $M$ on the average angular CRB for different routing-induced correlation levels. The results show that increasing the number of effective radiation points consistently improves the sensing accuracy for all considered correlation regimes. It is observed that the angular CRB decreases monotonically as $M$ increases. This behavior follows directly from the proposed Fisher information formulation. Specifically, a larger number of effective radiation points creates richer spatial sampling of the propagated sensing field and increases the dimensionality of the distributed sensing aperture. Consequently, the received signal becomes more sensitive to angular perturbations of the target, which increases the derivative energy of the conditional received signal with respect to the target angle. As a result, the Fisher information increases and the CRB decreases.

The figure also shows that routing-induced correlation has a significant impact on the achievable sensing performance. In the weak correlation regime, e.g., $\kappa=0.2$, the distributed radiation points contribute nearly independent sensing modes, leading to a large effective sensing diversity and substantially lower estimation error. In contrast, strong correlation, e.g., $\kappa=0.8$, causes the reradiated sensing components to become increasingly aligned, which reduces the effective sensing rank of the E-FAS aperture. Therefore, although increasing $M$ still improves the sensing accuracy under strong correlation, the achievable gain becomes considerably smaller compared with the weakly correlated regime. Moreover, the improvement rate gradually decreases for large values of $M$, particularly under moderate and strong correlation conditions. This indicates the existence of a diminishing return region, where adding additional radiation points provides only limited sensing diversity enhancement due to correlation induced mode coupling among the reradiated paths. Hence, the results in this figure strongly support the theoretical claim that the sensing capability of E-FAS architectures is governed not only by the number of reradiated sensing paths, but also by the effective spatial independence among the generated radiation modes.

\begin{figure}[t]
\centering
\includegraphics[width=0.85\columnwidth]{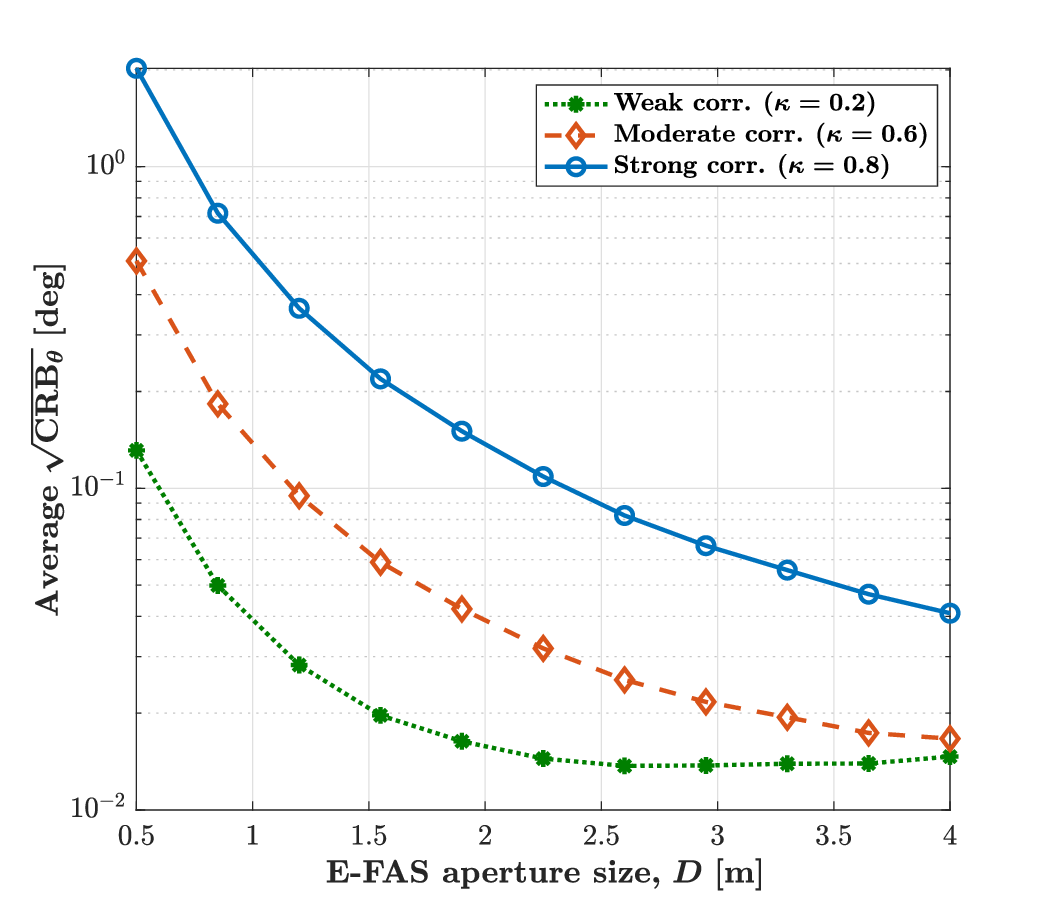}
\vspace{-2mm}\caption{Average CRB versus the E-FAS aperture size $D$ for different routing-induced correlation levels.}\label{fig:crb-d}
\vspace{-6mm}
\end{figure}

Now, Fig.~\ref{fig:crb-d} illustrates the impact of the E-FAS aperture size $D$ on the average angular sensing CRB for different routing-induced correlation levels. The results indicate that enlarging the E-FAS radiation aperture substantially improves the sensing accuracy, particularly under moderate and strong correlation conditions. As can be observed, the angular CRB decreases monotonically as the aperture size increases. This behavior is due to the spatial resolution capability of the distributed reradiation surface. A larger aperture generates wider spatial separation among the effective radiation points, which enhances the angular discrimination capability of the received sensing field. Consequently, the conditional received signal becomes more sensitive to small target angle variations, resulting in a larger Fisher information and a lower CRB.

On the other hand, the figure further shows that the performance improvement becomes more significant under strong routing-induced correlation. Specifically, when the reradiated sensing modes are highly correlated, increasing the aperture size effectively reduces the spatial coupling among the distributed radiation components. Therefore, aperture expansion partially restores the effective sensing diversity of the E-FAS architecture and mitigates sensing mode collapse. In contrast, under weak correlation conditions, the distributed radiation points are already sufficiently decorrelated even for moderate aperture sizes. As a result, the achievable diversity gain saturates after a certain aperture region, leading to a gradual flattening of the CRB curve for large values of $D$. This trend indicates that once the effective sensing manifold becomes sufficiently spatially resolved, further aperture expansion provides only limited additional Fisher information.

\begin{figure}[t]
\centering
\includegraphics[width=0.85\columnwidth]{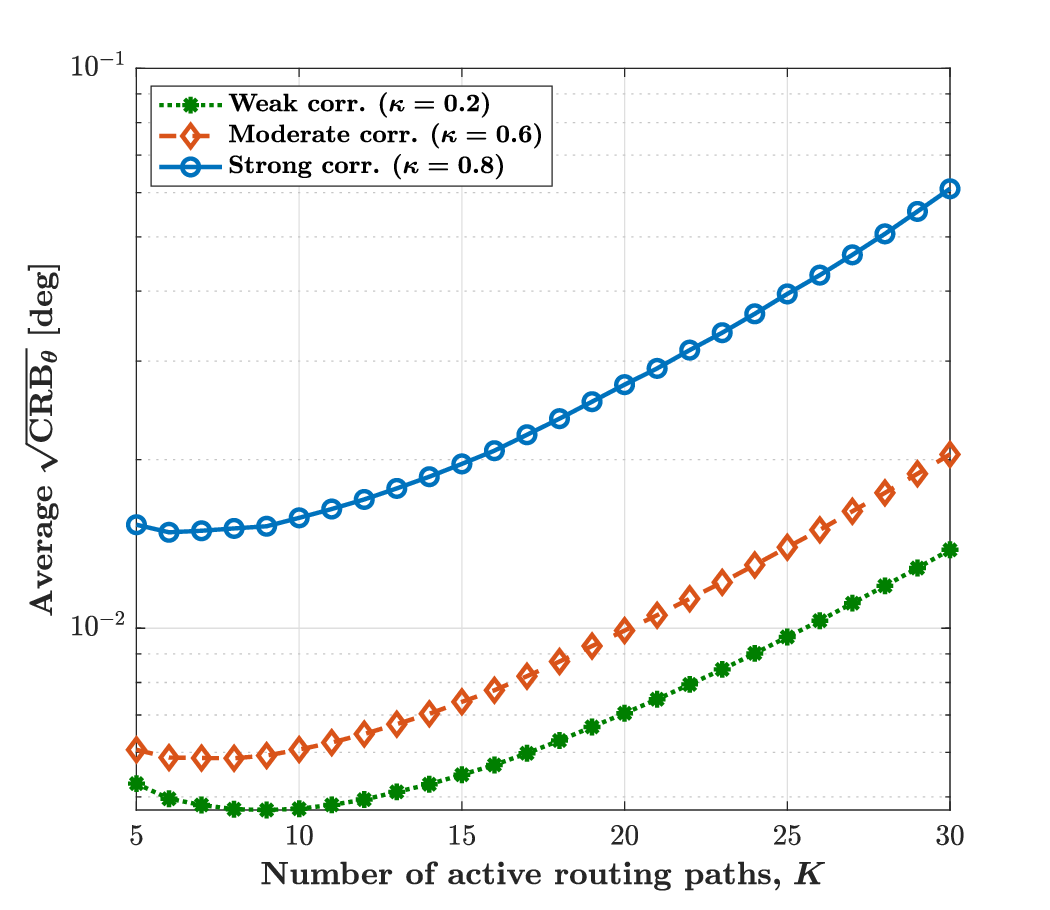}
\vspace{-2mm}\caption{Average CRB versus the number of active routing paths $K$  for different routing-induced correlation levels.}\label{fig:crb-k}
\vspace{-6mm}
\end{figure}

Fig.~\ref{fig:crb-k} shows the gain-diversity trade-off of the proposed E-FAS sensing architecture as a function of the number of active routing paths $K$ for different routing-induced correlation levels. The figure shows the average angular CRB achieved under different routing activation regimes. It is observed that the angular CRB initially decreases slightly for small values of $K$ and then increases progressively as the number of active routing paths becomes large. This behavior demonstrates the fundamental gain-diversity trade-off introduced by the proposed routing architecture. For small and moderate values of $K$, activating additional routing paths increases the spatial sensing diversity of the reradiated field and enriches the effective sensing manifold. Consequently, the received signal becomes more sensitive to target angle variations, which improves the Fisher information and reduces the CRB.

However, when the number of active routing paths becomes large, the routing power budget must be distributed among an increasing number of reradiated components. As a result, the coherent sensing gain contributed by each individual path decreases. Moreover, the additional routing paths become increasingly correlated, particularly under moderate and strong correlation regimes. This effect reduces the effective sensing rank of the distributed aperture and weakens the angular sensitivity of the received signal. Therefore, beyond a certain operating point, activating more routing paths degrades the sensing performance and increases the CRB. The figure further shows that the degradation becomes much more severe under strong routing-induced correlation conditions. Specifically, when $\kappa$ is large, the reradiated sensing modes become highly aligned, causing significant sensing mode collapse. Consequently, the achievable diversity gain provided by increasing $K$ becomes very limited, while the routing power splitting loss continues to increase. These results confirm the theoretical claim that the sensing performance of E-FAS-enabled ISAC systems is governed by a nontrivial trade-off between coherent routing gain and effective sensing diversity. Therefore, the number of active routing paths must be carefully optimized according to the underlying routing correlation conditions.

\begin{figure}[t]
\centering
\includegraphics[width=0.85\columnwidth]{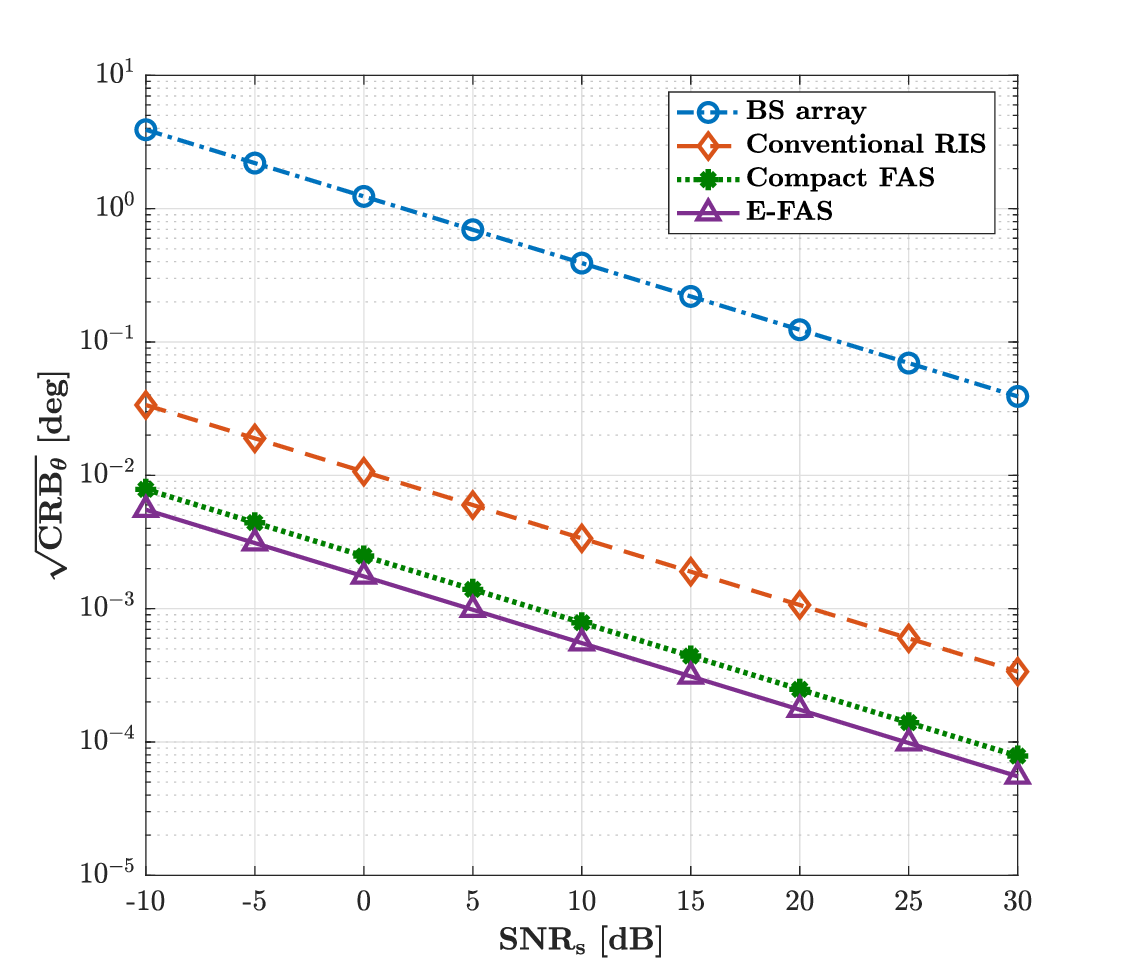}
\vspace{-2mm}\caption{Benchmark comparison of angular sensing performance versus sensing SNR $\mathrm{SNR}_{\mathrm s}$ for different sensing architectures.}\label{fig:crb-snrb}
\vspace{-6mm}
\end{figure}

Fig.~\ref{fig:crb-snrb} compares the sensing performance of the proposed E-FAS with several benchmarks, including a conventional BS array, a conventional RIS, and a compact FAS architecture. For a fair comparison, all schemes employ the same carrier frequency, BS antenna configuration, transmit sensing power, target geometry, number of effective radiation points $M=16$, and routing-induced correlation level $\kappa=0.45$. The RIS benchmark is modeled as a compact reradiation surface with aperture size $D=0.8~\mathrm{m}$, while the compact FAS benchmark employs a moderate aperture of $D=1.5~\mathrm{m}$. In contrast, the proposed E-FAS architecture utilizes a larger distributed aperture of $D=3~\mathrm{m}$. The selected aperture sizes are intended to reflect the characteristic spatial scales of the considered architectures. Conventional RIS deployments are generally constrained to compact reradiation surfaces, whereas FAS-based architectures enable larger effective apertures through distributed radiation points. The proposed E-FAS further exploits routing-enabled distributed reradiation to realize a substantially larger sensing aperture. To isolate the impact of sensing geometry and aperture diversity, the benchmark comparison is performed under geometry-normalized channels, where all architectures are evaluated under the same reference sensing SNR. It is observed that all considered architectures exhibit decreasing angular CRB as the sensing SNR increases, which is consistent with the inverse dependence of the Fisher information on the noise variance. However, substantial performance differences appear among the considered architectures due to their different effective sensing apertures and spatial sensing capabilities. The conventional BS array exhibits the highest CRB over the entire SNR range since its sensing capability is fundamentally limited by the physical aperture of the colocated antenna array. Without distributed reradiation support, the received sensing field provides limited angular diversity and weak spatial resolution. The conventional RIS significantly improves the sensing accuracy compared with the BS array by introducing reradiated sensing components. Nevertheless, its performance remains constrained by the compact reradiation aperture, which limits the achievable spatial discrimination capability. The compact FAS architecture further improves the sensing performance by increasing the effective reradiation aperture. The larger spatial separation among the radiation points enhances the sensitivity of the received signal to target-angle variations and consequently increases the Fisher information. The proposed E-FAS architecture achieves the lowest CRB among the considered configurations due to its substantially larger distributed aperture. The enlarged aperture creates a richer sensing manifold and provides stronger angular discrimination capability, leading to a more informative received sensing field and improved estimation accuracy.  

\begin{figure}[t]
\centering
\includegraphics[width=0.85\columnwidth]{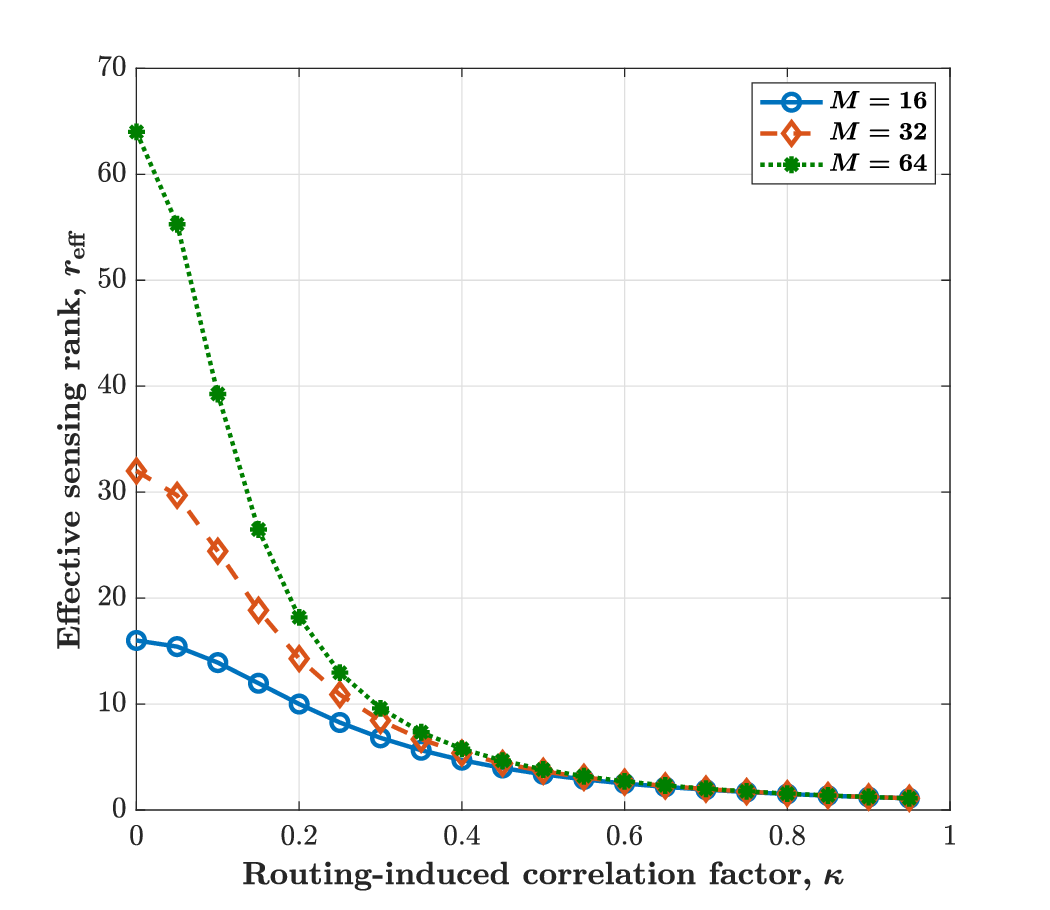}
\vspace{-2mm}\caption{Routing-induced sensing mode collapse versus the routing-induced correlation factor $\kappa$ for different numbers of effective radiation points $M$.}\label{fig:r-kappa}
\vspace{-6mm}
\end{figure}

To further visualize the impact of routing-induced correlation on the effective sensing dimensionality, we define the effective sensing rank as $r_{\mathrm{eff}}=\frac{\left(\mathrm{tr}(\mathbf R_{\mathrm s})\right)^2}{\mathrm{tr}\left(\mathbf R_{\mathrm s}^2\right)}$, where $\mathbf R_{\mathrm s}$ denotes the routing-dependent sensing covariance matrix. The effective sensing rank provides a quantitative measure of the number of effectively independent sensing modes supported by the distributed reradiation architecture. Fig. \ref{fig:r-kappa} illustrates the impact of routing-induced correlation on the effective sensing rank of the proposed E-FAS sensing architecture for different numbers of effective radiation points $M$. The results in this figure illustrate how increasing spatial correlation among reradiated sensing components gradually reduces the effective dimensionality of the sensing field. It is observed that the effective sensing rank decreases monotonically as the routing-induced correlation factor $\kappa$ increases. This behavior directly follows from the progressive alignment of the reradiated sensing modes. When the routing correlation is weak, the radiation points generate nearly independent sensing components, resulting in a high-dimensional sensing manifold and large effective sensing rank. In this regime, increasing the number of effective radiation points substantially enriches the spatial DoFs of the propagated sensing field. However, as $\kappa$ increases, the reradiated sensing modes become increasingly correlated and begin to collapse into a smaller number of dominant spatial components. Consequently, many reradiated paths contribute highly redundant sensing information, leading to a rapid reduction in the effective sensing rank. This effect becomes particularly severe in the strong correlation regime, where the effective sensing ranks corresponding to different values of $M$ converge toward similar low-dimensional operating points.

Furthermore, this figure demonstrates that larger values of $M$ provide higher effective sensing rank under weak and moderate correlation conditions. Nevertheless, this advantage gradually diminishes as the correlation strength increases. This observation once again confirms that simply increasing the number of reradiated sensing paths is insufficient to guarantee high sensing diversity unless adequate spatial independence is maintained among the generated reradiation modes. 

\vspace{-4mm}
\section{Conclusion}\label{sec:con}
This paper developed a routing-dependent analytical framework for E-FAS-enabled ISAC systems by explicitly modeling the bidirectional propagation structure created by distributed reradiation and SW routing. Based on the proposed conditional sensing model, closed-form Fisher information expressions and the angular CRBs were derived to characterize the fundamental sensing limits of E-FAS architectures. The analysis demonstrated that the sensing performance of E-FAS systems is fundamentally governed by the interplay among SW routing gain, routing-induced correlation, and distributed aperture geometry. In particular, increasing the coherent routing gain was shown to improve sensing accuracy, while strong routing-induced correlation reduces the effective sensing dimensionality and degrades angular estimation performance. The results further revealed that enlarging the distributed E-FAS aperture and increasing the number of effective radiation points can substantianlly enhance sensing resolution by improving spatial diversity. This paper also identified a fundamental gain-diversity trade-off in E-FAS-enabled sensing. Activating more routing paths improves spatial sensing diversity only up to a certain operating region, beyond which routing power splitting and mode correlation begin to dominate the sensing behavior. Numerical results verified the accuracy of the proposed analytical framework and demonstrated close agreement between the derived CRB expressions and Monte Carlo simulations under various routing and propagation conditions.\vspace{-5mm}

\bibliographystyle{IEEEtran}

\end{document}